\documentclass[review,onethmnum]{siamart220329}

\usepackage[T1]{fontenc}

\usepackage{amsmath,amssymb}
\usepackage{bbm}
\usepackage{mathtools}

\usepackage{graphicx}
\usepackage[table]{xcolor}
\usepackage{tikz}
\usetikzlibrary{shapes.geometric, fit, positioning}
\usetikzlibrary{fit, backgrounds}

\usepackage{array} 
\newcolumntype{C}{>{\centering\arraybackslash}p{1.2em}}
\usepackage{booktabs}
\usepackage{makecell}
\usepackage{nicematrix}
\definecolor{myblue}{HTML}{B4C6E7}

\newcommand{\clg}[1]{\cellcolor{lightgray}#1}

\usepackage{enumitem}
\usepackage{subcaption}

\makeatletter
\@ifundefined{algorithm}{}{ }
\@ifundefined{listofalgorithms}{}{}
\makeatother
\usepackage[ruled,vlined]{algorithm2e}

\usepackage[hidelinks]{hyperref}
\usepackage{cleveref}
\usepackage{url}

\usepackage{comment}
\excludecomment{outline}
\usepackage{xspace}
\usepackage{dirtytalk}
\usepackage{cancel}
\usepackage{orcidlink}

\makeatletter
\@ifundefined{abstract}{%
  \newenvironment{abstract}{\section*{Abstract}\noindent}{\par}%
}{}
\@ifundefined{remark}{\newtheorem{remark}{Remark}}{}
\makeatother

\crefname{lemma}{Lemma}{Lemmas}
\crefname{figure}{Figure}{Figures}
\crefname{problem}{Problem}{Problems}
\crefname{proposition}{Proposition}{Propositions}
\crefname{theorem}{Theorem}{Theorems}
\crefname{remark}{Remark}{Remarks}
\crefname{algorithm}{Algorithm}{Algorithms}
\crefname{example}{Example}{Examples}
\crefname{claim}{Claim}{Claim}
\crefname{table}{Table}{Tables}
\crefname{corollary}{Corollary}{Corollaries}
\crefname{definition}{Definition}{Definitions}
\crefname{equation}{Equation}{Equations}
\crefname{section}{Section}{Sections}

\renewcommand{\thesection}{\arabic{section}}
\renewcommand{\thesubsection}{\thesection.\arabic{subsection}}
\renewcommand{\thesubsubsection}{\thesubsection.\arabic{subsubsection}}

\normalsize
\renewcommand{\arraystretch}{1.5}

\renewcommand{\theenumii}{\theenumi.\arabic{enumii}.}

\newcommand{\ceil}[1]{\left\lceil {#1} \right\rceil}
\newcommand{\floor}[1]{\left\lfloor {#1} \right\rfloor}

\DeclareMathOperator{\Vol}{Vol}
\DeclareMathOperator{\supp}{supp}

\let\emptyset\varnothing

\newcommand{\E}{\mathbb{E}}

\providecommand{\headers}[2]{}%
\headers{The Minimum Number of Samples in the Labeled Coupon Collector Problem}{A. Tan, D.\ Bar-Lev, R.\ Gabrys, Z.\ Yakhini, and P.\ H.\ Siegel}

\begin{document}

\title{The Minimum Number of Samples in the Labeled Coupon Collector Problem}

\author{Andrew Tan\and Daniella Bar-Lev\and Ryan Gabrys\and Zohar Yakhini\and Paul H. Siegel%
\thanks{The work of Z. Yakhini was funded by the European Union (DiDAX, 101115134). Views and opinions expressed are however those of the author(s) only and do not necessarily reflect those of the European Union or the European Research Council Executive Agency. Neither the European Union nor the granting authority can be held responsible for them. The work of D. Bar-Lev was supported in part by Schmidt Sciences. The work of D. Bar-Lev, R. Gabrys, A. Tan, and P. H. Siegel was supported in part by NSF Grant CCF-2212437. (Andrew Tan and Daniella Bar-Lev contributed equally to this work.)}\textsuperscript{,}%
\thanks{Part of this work was presented and published at the IEEE Information Theory Workshop (ITW), Sydney, Australia, 2025~\cite{Tan_ITW_2025}.}}

\maketitle

\begin{outline}
    \section{Abstract}
\end{outline}
\begin{abstract}
The Labeled Coupon Collector Problem (LCCP) is an extension of the classical Coupon Collector Problem (CCP) in which one must deduce an unknown mapping between coupons and labels by collecting unordered subsets of $k$ coupons and their associated labels, where deduction relies on analyzing the intersections and differences of the unordered subsets.
In this work, we establish a fundamental connection between the LCCP and the combinatorial theory of separating systems. Using this framework, we completely characterize the minimum required number of samples to deduce all labels in two settings: Type-I, where the label set is known in advance, and Type-II, where it is unknown. We provide deterministic constructions of minimal separating and covering systems, yielding explicit bounds that are tighter than those found in existing literature. Furthermore, we extend our analysis to the partial recovery problem, characterizing the minimum number of samples required to identify a specified number of target coupons.
\end{abstract}

\section{Introduction}
\label{sec:introduction}
The Coupon Collector Problem (CCP) \cite{Feller},\cite{Flajolet} is a classical combinatorial problem dating back to De Moivre, Euler, and Laplace, which can be described as follows. Consider a promotional contest in which a prize is won by collecting at least one of each type of coupon distributed in packages of a certain product. Assuming the number of coupon types is a finite number $n$ and that each package contains exactly one coupon, opening a package can be viewed as sampling one coupon independently and uniformly at random with replacement from the set of possible coupon types. The CCP investigates the number of samples, $N$,  needed to collect all $n$ coupon types. Specific problems of interest include finding the minimum, expectation, and other properties of the distribution of $N$. The minimum is trivially seen to be $n$. The sample counts between collecting distinct coupon types are independent geometric random variables; thus,  the expectation is found to be $\E[N]=nH_n$, where $H_n =  \sum_{j=1}^n \frac{1}{j}$ denotes the $n$-th harmonic number~\cite{Feller}. Formulas for higher moments of $N$ can be computed~\cite{LeonGarcia}, and limit theorems for the asymptotic distribution of $N$ have also been studied~\cite{ErdosRenyi}, \cite{BaumBillingsley}. 

A notable extension of the CCP is the group drawing scenario, corresponding to $k$ distinct coupons per package, or $k$ distinct coupons drawn per sample with replacement, for some fixed constant $k < n$. 
Properties of the number of samples, $N(k)$, needed to collect all coupon types and related quantities have been studied in,  e.g.,~\cite{BT22, JS10, Stadje, AR01, BerendSher}. The minimum number of samples  needed to collect all $n$ coupon types is trivially $\lceil \frac{n}{k}\rceil$. 

Several other variants and generalizations of the CCP have been proposed, and applications are found in a wide range of mathematical, scientific, and engineering disciplines. An extensive list of references of a fundamental and applied nature can be found in~\cite{BonehHofri}.  

In this paper, we introduce a new generalization of the CCP called the \textit{Labeled Coupon Collector Problem} (LCCP) that combines aspects of both collection and inference. Again using a  contest analogy, each of the $n$ coupon types is uniquely associated with a distinct tag or \textit{label}, and each package contains a subset of $k<n$ distinct coupons and their associated labels. More formally, 
we view each package as a sample containing two \textbf{unordered} sets: the $k$ distinct coupons drawn uniformly at random with replacement, and their $k$ associated labels.
If $k>1$, the association between coupon types and labels cannot be immediately deduced. To win the contest, one must correctly identify the label associated with each coupon type by combining the information gained from multiple packages.
We are interested in the number of samples, $Q(n,k)$, needed to establish the correspondence. In particular, the main objective of this paper is the investigation of the minimum possible value of $Q(n,k)$. Unlike the previously mentioned variants, analyzing the minimum in the current context is not trivial. 

We consider two settings: Type-I, where the label set is known in advance, and Type-II, where it is unknown, with corresponding number of samples $Q_\mathrm{I}(n,k)$ and $Q_\mathrm{II}(n,k)$.
To illustrate how labels are deduced in both settings, consider the following example, which is depicted graphically in Figure~\ref{fig:deduction}.
\begin{example}
    \label{ex1}
    Suppose there are $n = 4$ coupons $\{1,2,3,4\}$ with corresponding labels $\{A,B,C,D\}$, and we draw $k = 2$ coupons per sample. Suppose the first two samples are given by:
    \begin{enumerate}
        \item coupons $\{1,2\}$ and labels $\{B,A\}$, and
        \item coupons $\{2,3\}$ and labels $\{B,C\}$. 
    \end{enumerate}
    \noindent
    Since coupon $2$ appears in both samples, we deduce the correspondences $1 \leftrightarrow A, 2 \leftrightarrow B, 3 \leftrightarrow C$. In the Type-I setting, we further deduce via elimination that $4\leftrightarrow D$. 
     However, the correspondence $4\leftrightarrow D$ cannot be deduced from these two samples in the Type-II setting because the label $D$ has not been drawn and the set of labels $\{A,B,C,D\}$ is not known in advance. A third sample is required, such as: 
     \begin{enumerate}
        \item[3)] coupons $\{4,1\}$ and labels $\{A,D\}$.
    \end{enumerate}
\end{example}

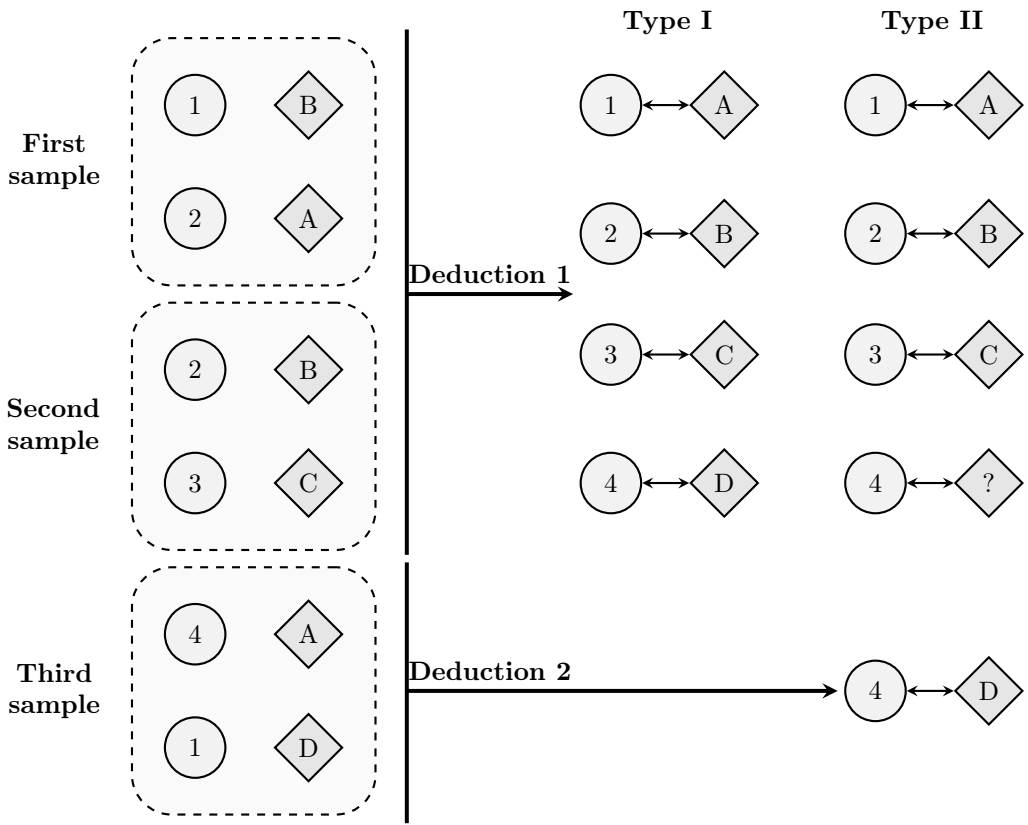
\begin{figure}[ht!]
    \centering
    \begin{tikzpicture}[
        >=stealth,
        coupon/.style={
            circle, draw, minimum size=8mm, thick, fill=black!5
        },
        labelnode/.style={
            diamond, draw, minimum size=9mm, thick,
            fill=black!10, inner sep=1pt
        },
        samplebox/.style={
            rectangle, draw, dashed, thick,
            rounded corners=15pt, inner sep=12pt, fill=black!2
        },
        deduction/.style={<->, thick}
    ]

    \node[coupon] (s1_c1) at (0, 4.5) {1};
    \node[coupon] (s1_c2) at (0, 3.0) {2};
    \node[labelnode] (s1_lB) at (1.5, 4.5) {B};
    \node[labelnode] (s1_lA) at (1.5, 3.0) {A};
    \scoped[on background layer]
        \node[samplebox,
              fit=(s1_c1) (s1_c2) (s1_lB) (s1_lA)] (box1) {};
    \node[left=8pt of box1, font=\bfseries, align=center]
        {First\\sample};

    \node[coupon] (s2_c2) at (0, 1.0) {2};
    \node[coupon] (s2_c3) at (0, -0.5) {3};
    \node[labelnode] (s2_lB) at (1.5, 1.0) {B};
    \node[labelnode] (s2_lC) at (1.5, -0.5) {C};
    \scoped[on background layer]
        \node[samplebox,
              fit=(s2_c2) (s2_c3) (s2_lB) (s2_lC)] (box2) {};
    \node[left=8pt of box2, font=\bfseries, align=center]
        {Second\\sample};

    \node[coupon] (s3_c4) at (0, -2.5) {4};
    \node[coupon] (s3_c1) at (0, -4.0) {1};
    \node[labelnode] (s3_lA) at (1.5, -2.5) {A};
    \node[labelnode] (s3_lD) at (1.5, -4.0) {D};
    \scoped[on background layer]
        \node[samplebox,
              fit=(s3_c4) (s3_c1) (s3_lA) (s3_lD)] (box3) {};
    \node[left=8pt of box3, font=\bfseries, align=center]
        {Third\\sample};

    \node[font=\bfseries] at (6.25, 5.6) {Type I};
    \node[font=\bfseries] at (9.75, 5.6) {Type II};

    \draw[line width=1.5pt]
        (2.8, 5.5) -- (2.8, -1.45);
    \draw[->, line width=1.5pt]
        (2.8, 2)
        -- node[above, font=\bfseries] {Deduction 1} (5.0, 2);

    \foreach \couponnum/\labeltext/\ypos in
        {1/A/4.5, 2/B/2.8, 3/C/1.2, 4/D/-0.5} {
        \node[coupon] (t1_c\couponnum)
            at (5.5, \ypos) {\couponnum};
        \node[labelnode] (t1_l\couponnum)
            at (7, \ypos) {\labeltext};
        \draw[deduction]
            (t1_c\couponnum) -- (t1_l\couponnum);
    }

    \foreach \couponnum/\labeltext/\ypos in
        {1/A/4.5, 2/B/2.8, 3/C/1.2, 4/?/-0.5} {
        \node[coupon] (t2_c\couponnum)
            at (9, \ypos) {\couponnum};
        \node[labelnode] (t2_l\couponnum)
            at (10.5, \ypos) {\labeltext};
        \draw[deduction]
            (t2_c\couponnum) -- (t2_l\couponnum);
    }

    \draw[line width=1.5pt]
        (2.8, -1.55) -- (2.8, -5);
    \draw[->, line width=1.5pt]
        (2.8, -3.25)
        -- node[above, font=\bfseries] {Deduction 2} (5.0, -3.25)
        -- (8.5, -3.25);

    \node[coupon] (t2_final_c4) at (9, -3.25) {4};
    \node[labelnode] (t2_final_lD) at (10.5, -3.25) {D};
    \draw[deduction] (t2_final_c4) -- (t2_final_lD);

    \end{tikzpicture}
    \caption[Deduction of labels as described by Example]
        {Deduction of labels as described by
         Example~\protect\ref{ex1}. For detailed explanation, see Example~\protect\ref{ex1} text.}
    \label{fig:deduction}
\end{figure}

As mentioned, in contrast to the CCP and the group drawing CCP, finding the minimum number of samples, $a_\mathrm{I}(n,k)$ and $a_\mathrm{II}(n,k)$, needed to deduce the labels associated with the $n$ coupon types in the Type-I and Type-II settings of the LCCP, respectively,  is not trivial for $k\geq 2$. This is the focus of this work.\footnote{We do not consider the expectation here. In~\cite{Tan_ITW_2025}, we conjectured that 
$\E[Q_\mathrm{II}(n,2)] - \E[Q_\mathrm{I}(n,2)] = \frac{1}{2}n + o(n)$ and $\E[Q_\mathrm{II}(n,2)] = \frac{1}{2}nH_n +o(n)$. These conjectures were confirmed in~\cite{BarakBerend}.}
We demonstrate that this task is mathematically equivalent to finding minimal separating systems, a concept previously studied by R\'enyi~\cite{Renyi61,Renyi62} and Katona~\cite{Katona}. Building upon their foundational work, we provide a full characterization of $a_\mathrm{I}(n,k)$ and $a_\mathrm{II}(n,k)$, develop explicit combinatorial constructions for sequences of samples that achieve the minimum, and establish tight bounds that improve upon the corresponding classical results for separating systems~\cite{Katona,Wegener}. We also extend our framework to the \textit{partial recovery setting}, where the labeling of only a  subset of $r$ coupons needs to be identified.

In addition to the mathematical allure of these questions, they can also be relevant in practical applications. For example, consider the ``traffic analysis problem'', in which the objective is to keep the identities of pairs of senders and recipients in a communication network confidential.  First defined in~\cite{Chaum1991}, a Chaumian threshold mix is an anonymity protocol designed to unlink sender-recipient pairs by collecting messages and only forwarding them once a predetermined number (the threshold) has been received. It does so through a mix node, which collects incoming packets and flushes them in batches of size $k$ to their receivers, intentionally stripping explicit routing headers \cite{Berthold2000, vanDenHooff2015}. A passive adversary effectively has access only to the unordered sets of senders and recipients in the mix node. 
Such an adversary could theoretically reconstruct the sender-recipient pairs by analyzing multiple batch flushes over time in an intersection attack \cite{Mathewson2004} similar to what is done in Example~\ref{ex1}. 
The LCCP serves as an abstraction of this problem that could help to analyze the effectiveness of this protocol. 
It should be noted that real mixes face users who send at different rates, one-to-many communication patterns, and dummy/cover traffic. Hence, the LCCP model offers a fundamental-limits idealization of the threshold-mix base case, rather than a direct model of any deployed system.
Nevertheless, as also discussed in Section V, some of these scenarios can be addressed by generalizations of the LCCP model. 


The remainder of this paper is organized as follows. Section II introduces the formal definition of the LCCP and rigorously establishes its connection to Katona's separating systems. Section III completely characterizes the minimum number of samples for full recovery and provides new bounds on the minimum that improve upon  those obtained from previous bounds on separating systems. Section IV extends these deterministic constructions to the partial recovery problem. Section V summarizes the results and points to open directions.  
%
\section{Preliminaries}

In this section, we define the LCCP and establish the connection to separating systems in \cite{Katona,Renyi61,Renyi62}.

\subsection{Problem Statement and Notations}

Let $[n]:= \{1,2,...,n\}$ be a set of coupons, $L$ be an alphabet of labels with $|L| = n$, and $f: [n] \rightarrow L$ be a bijective \textit{labeling function} assigning each coupon $j\in [n]$ a unique label $f(j) \in L$. For a fixed sample size $k \in [n-1]$, let $S_i\subseteq[n]$ be the $i$-th sampled subset of coupons, with labels given by $f(S_i) = \{f(j): j\in S_i\}$, where $|S_i| = k$. The sampled subsets $S_i$ and their corresponding labels $f(S_i)$ are given as unordered sets for all $i \geq 1$.

We define the random variable $Q(n,k)$ as the \textit{first} sample in which $f$ can be determined, i.e., each coupon's label can be deduced via set operations on the system of sampled pairs $$\{(S_i, f(S_i)) \}_{i=1}^{Q(n,k)}  = \{(S_1,f(S_1)), \dots, (S_{Q(n,k)},f(S_{Q(n,k)}))\}.$$

We consider this problem in two settings: one where the label alphabet $L$ is known at the start, and one where $L$ is unknown; we refer to the corresponding LCCP problems as \textit{Type-I} and \textit{Type-II} respectively. Now, we can define the LCCP.
\vspace{5pt}

\noindent
{\bf Problem [$(n,k)$-LCCP].}
For a fixed positive integers $n,k\in \mathbb{Z}^+$ such that $n > k \geq 1$, find the:
\begin{enumerate}
    \item[(1)] \textit{Minimum} of $Q(n,k)$, namely, the smallest possible value that $Q(n,k)$ can take on. This is denoted $$a_\mathrm{I}(n,k) =   \min \{ q: q\in \supp Q_\text{I}(n,k) \},$$ $$a_\mathrm{II}(n,k) =   \min \{ q: q\in \supp Q_\text{II}(n,k) \}$$ for Type-I and Type-II, respectively.
    
    \item[(2)] \textit{Expectation} of $Q(n,k)$

    \item[(3)] \textit{Variance} of $Q(n,k)$

    \item[(4)] \textit{Distribution} of $Q(n,k)$

\end{enumerate}

\vspace{5pt}

In this work, we address item (1), the (deterministic) minimum number of samples.

\subsection{Connection to Separating Systems}

Now, we give exact conditions on when a coupon's label can be determined. We first define the concept of separability in a finite set.

\begin{definition}
\label{def:sep}
    Given $n,m \in \mathbb{Z}^+$ and subsets $\{S_i\}_{i=1}^m \subseteq [n]$, we say that:
    \begin{itemize}
        \item $S_i$ \textit{separates} two distinct coupons $j, j' \in [n]$ if exactly one of them is contained in $S_i$. 
        \item $j$ is \textit{separable} given $\{S_i\}_{i=1}^m$ if for every $j' \neq j$, at least one $S_i$ separates $j$ and $j'$. Moreover, if every coupon is separable given $\{S_i\}_{i=1}^m$, then $\{S_i\}_{i=1}^m$ is a \textit{separating system} for $[n]$.
        \item $\{S_i\}_{i=1}^m$ is a \textit{covering system} (or \textit{covers} $[n]$) if every coupon $j \in [n]$ is contained in at least one sample $S_i$ for some $i \in [m]$. 
    \end{itemize}
\end{definition}

\noindent Note that the definition of a separating system is identical to the one in \cite{Katona}. We also define the concept of an \textit{index set}.

\begin{definition}
\label{def:index}
    Given coupon $j\in [n]$ and a system of $m$ subsets $\{S_i\}_{i=1}^m$, we define:
    \begin{itemize}
        \item the \textit{index set} of $j$ as the set $I(j) = \{ i\in [m]: j\in S_i \}$, 
        \item the preorder $\preceq$ and equivalence relation $\simeq$ on $[n]$ such that $j_1 \preceq j_2$ if $I(j_1) \subseteq I(j_2)$ and $j_1 \simeq j_2$ if $I(j_1) = I(j_2)$, and
        \item the \textit{separation class} of $j$ as the set $c(j)$ obtained by the following intersection of sets:
        \begin{equation}
      \label{candidates}
     c(j) = \left( \bigcap\limits_{i \in I(j)} S_i\right) \cap \left( \bigcap\limits_{i \in I(j)^{\mathrm{c}}} S_i^\mathrm{c}\right), 
    \end{equation}
    where $S_i^{\mathrm{c}}$ is the set complement of $S_i$ in $[n]$, and where we adopt the convention that the intersection over an empty index set yields the universal set $[n]$.
    \end{itemize}
\end{definition}

It is clear that $\preceq$ is a preorder (reflexive and transitive) and that $\simeq$ is an equivalence relation since they inherit these properties from set inclusion and set equality, respectively. Note that $\preceq$ becomes a partial order if and only if $\{S_i\}_{i=1}^m$ is a separating system, as antisymmetry is satisfied only when all index sets are distinct.

We have the following properties that establish an equivalent characterization of separability.

\begin{lemma}
    \label{lemma_separable}
    Given $n  \in \mathbb{Z}^+$ and a system of subsets $\{S_i\}_{i=1}^m$:
    \begin{enumerate}
        \item [(a)] $\{S_i\}_{i=1}^m$ separates $j_1$ and $j_2$ if and only if $j_1 \not\simeq j_2$.
        \item [(b)] $c(j)$ is the smallest set containing $j$ that can be generated by set operations on $\{S_i\}_{i=1}^m$.
        \item [(c)] $j_1 \simeq j_2$ if and only if $c(j_1) = c(j_2)$. Consequently, separation classes are the equivalence classes with respect to $\simeq$.
        \item [(d)] $c(j) = \{j\}$ if and only if $\{j\}$ is separable. Consequently, this holds for all $j \in [n]$ if and only if $\{S_i\}_{i=1}^m$ is a separating system.
    \end{enumerate}
\end{lemma}

\begin{proof}
    We prove each part directly.
    \begin{enumerate}
        \item[(a)] $\{S_i\}_{i=1}^m$ separates $j_1$ and $j_2$ if and only if there is an $S_i$ that contains either $j_1$ or $j_2$ but not both. Equivalently, $i$ is contained in either $I(j_1)$ or $I(j_2)$ but not both, which holds if and only if $I(j_1) \neq I(j_2)$, i.e., $j_1 \not\simeq j_2$.
        \item [(b)] By definition, if $j\notin S_i$, then $j\in S_i^\mathrm{c}$ so $j\in c(j)$. Observe that any set generated via set operations on $\{S_i\}_{i=1}^m$ can be written as a disjoint union of sets  in $\mathcal{A} = \{\bigcap_{i=1}^m A_i : A_i \in \{S_i, S_i^\mathrm{c} \} \}$. Because $c(j) \in \mathcal{A}$ is a set containing $j$, any set that contains $j$ must also contain $c(j)$. Hence, $c(j)$ is the smallest set containing $j$ that can be generated via set operations.
        \item [(c)] If $j_1 \simeq j_2$, then $I(j_1) = I(j_2)$ and $I(j_1)^\mathrm{c} =I(j_2)^\mathrm{c} $, which means that $c(j_1) = c(j_2)$. Conversely, if $c(j_1) = c(j_2)$, then from part (b), $j_1 \in c(j_2)$, which means that $j_1 \in S_i$ whenever $i \in I(j_2)$, and so $I(j_1) \subseteq I(j_2)$. Repeating this argument for $j_2 \in c(j_1)$ will give us $I(j_1) = I(j_2)$ and $j_1 \simeq j_2$.
        \item [(d)] Since $j \in c(j)$, then by part (c), $c(j) = \{j\}$ if and only if $j \not\simeq j'$ for all $j' \in [n] \setminus \{j\}$, and by part (a), this holds if and only if $\{S_i\}_{i=1}^m$ separates $j$ and $j'$ for all $j' \neq j$, or equivalently that $j$ is separable.
    \end{enumerate}
\end{proof}

For coupon $j \in [n]$, its separation class $c(j)$ can be viewed as the set of coupons that cannot be separated from $j$ given $\{S_i\}_{i=1}^m$. Thus, \cref{candidates} establishes a canonical representation for each singleton if $\{S_i\}_{i=1}^m$ is a separating system, allowing us to establish \cref{thm_deducing}.

\begin{theorem}
\label{thm_deducing} 
    For any $n \in \mathbb{Z}^+$, $\{S_i\}_{i=1}^m \subseteq [n]$, the label of every coupon can be determined given $\{(S_i, f(S_i))\}_{i=1}^m$ if and only if $\{S_i\}_{i=1}^m$ is a separating system (additionally a covering system in the Type-II setting).
\end{theorem}

\begin{proof}
    In the Type-II setting, determining every coupon's label requires every label to be observed at least once. This implies $\bigcup_{i=1}^m f(S_i) = L = f([n]) $, so by the bijectivity of $f$, $ \bigcup_{i=1}^m S_i =[n]$. Therefore, $\{S_i\}_{i=1}^m$ being a covering system is a necessary condition for determining every coupon's label in the Type-II setting.
    
    Next, we show that if every coupon's label can be determined given $\{S_i\}_{i=1}^m$, then $\{S_i\}_{i=1}^m$ is a separating system. As in the proof of \cref{lemma_separable}, the smallest set of labels containing $f(j)$ that can be obtained via set operations on $\{f(S_i)\}_{i=1}^m$ is:
    \begin{equation}
        \label{eqn:deducing}
        c_f(j):=\left( \bigcap_{i\in I(j)}f(S_i) \right) \cap \left( \bigcap_{i\in I(j)^\mathrm{c}}f(S_i)^\mathrm{c} \right).
    \end{equation}
    Therefore, if the label of every $j\in [n]$ can be determined, we must have $c_f(j) = \{f(j)\}$. Moreover, $\{f(j)\}$ must be the smallest set of labels containing $f(j)$ that can be obtained in this manner. By bijectivity and the definition of $c(j)$, we obtain $\{f(j)\} = f(c(j))$ and $\{j\} = c(j)$ for all $j \in [n]$, which means $\{S_i\}_{i=1}^m$ is a separating system.
    
    Conversely, if $\{S_i\}_{i=1}^m$ is a separating system (additionally a covering system in the Type-II setting), then by \cref{lemma_separable}, the label of coupon $j \in [n]$ is the sole element of $c_f(j)$.
\end{proof}

\begin{remark}
\label{remark_deducing}
    Thus, the minimum number of samples required to deduce the label of all $n$ coupons is mathematically equivalent to the minimum number of subsets required to form a separating system for $[n]$ (additionally a covering system in the Type-II setting).
    \hfill $\square$
\end{remark}

Finally, we define the draw matrix, which was also utilized by~\cite{Katona}.

\begin{definition}
    \label{def_draw_matrix}
    Given $n,m \in \mathbb{Z}^+$ and subsets $\{S_i\}_{i=1}^m \subseteq [n]$, the \textit{draw matrix} corresponding to $\{S_i\}_{i=1}^m$ is a binary matrix $M \in \mathbb{F}_2^{m\times n}$ such that $M_{ij} = 1$ when $j \in S_i$.
    \hfill $\square$
\end{definition}

\begin{remark}
    In the $(n,k)$-LCCP, each row of $M$ has Hamming weight $k$ and represents a sampled subset, and each column corresponds to a coupon. Therefore, it is clear that $\{S_i\}_{i=1}^m$ being a separating system is equivalent to the following conditions:
    \begin{itemize}
        \item the index sets $\{I(j)\}_{j=1}^n$ are distinct,
        \item the columns of $M$ are distinct.
    \end{itemize}
    Similarly, $\{S_i\}_{i=1}^m$ being a covering system is equivalent to the following:
    \begin{itemize}
        \item the index sets $\{I(j)\}_{j=1}^n$ are all nonempty, 
        \item the columns of $M$ are all nonzero (i.e., have nonzero weight).
    \end{itemize}

\end{remark}

\subsection{Deducing Labels from Separating Systems}

Suppose we can construct a separating system and now want to use it to determine the label of each coupon. This can be done directly using \cref{eqn:deducing} by performing $m-1$ set intersections. However, this is inefficient, as we can exploit the fact that $\preceq$ from \cref{def:index} is a partial order on $[n]$ given a separating system. Namely, the coupons have distinct index sets and thus are never equivalent with respect to $\preceq$.

We focus our attention on \textit{maximal} elements with respect to $\preceq$, i.e., coupons whose index sets are not proper subsets of any other index set. Clearly, given a separating system, coupons that appear most frequently in the samples will be maximal elements with respect to $\preceq$.

\begin{lemma}
    \label{lem:freq}
    Given $\{(S_i, f(S_i) )\}_{i=1}^m$ where $\{S_i\}_{i=1}^m$ is a separating system, let $j^*\in [n]$ such that $|I(j^*)| =  \max_{j\in [n]} |I(j)|$. Then:
    \begin{align}
        \{j^* \} = \bigcap_{i \in I(j^*)} S_i, \quad \text{and} \\
        \{f(j^*)\} = \bigcap_{i \in I(j^*)} f(S_i).
    \end{align}
\end{lemma}

\begin{proof}
    Suppose we have $j \in [n]$ such that $j \in \bigcap_{i \in I(j^*)} S_i$, which holds if and only if $j \in S_i$ for all $i \in I(j^*)$. This implies $I(j) \supseteq I(j^*)$ and $|I(j)| \geq |I(j^*)|$, but since $j^*$ is defined to have maximum frequency, then we must have $|I(j)| = |I(j^*)|$ and $I(j) = I(j^*)$. Moreover, since $\{S_i\}_{i=1}^m$ is a separating system, every index set is distinct, so this equality holds if and only if $j = j^*$.
\end{proof}

\begin{remark}
    By applying \cref{lem:freq}, we can obtain the label $f(j^*)$ of a maximal element $j^*$. Because the index sets of the remaining coupons are static, their frequencies and relative ordering with respect to $\preceq$ never change. Thus, we can simply determine $f(j^*)$, extract it from the pool of unknown labels by updating $f(S_i) \gets f(S_i) \setminus \{ f(j^*)\}$ for all $i \in I(j^*)$, and proceed to the next maximal coupon. By induction, this process determines the label of every coupon, as outlined in \cref{alg:peeling}. This significantly reduces the number of set operations compared to applying \cref{eqn:deducing} directly.
\end{remark}

\begin{algorithm}
\DontPrintSemicolon
\caption{Greedy Intersection--Peeling Algorithm}
\label{alg:peeling}
\KwIn{Separating system of sampled subsets $\{(S_1,f(S_1)),\dots,(S_m,f(S_m))\}$.}
\KwOut{All coupon-label correspondences $(j,f(j))$ extracted from intersections of the sets in $\{f(S_i)\}_{i=1}^m$.}

\BlankLine
\textbf{Preprocessing:}\;
Let $\mathcal{U} \gets [n]$ be the set of coupons not yet extracted.\;
\For{$j \in \mathcal{U}$}{
    Compute and store the index set $I(j) := \{\, i \in [m] : j \in S_i \,\}$.\;
    Compute and store the frequency $F(j) := |I(j)|$.\;
}
Sort the coupons in $\mathcal{U}$ in decreasing order of $F(j)$.\;
\BlankLine
\While{$\mathcal{U} \neq \emptyset$}{
    Pop the first element $j$ from the sorted $\mathcal{U}$ (which has maximum frequency).\;
    Compute the intersection $c_f(j) = \bigcap_{i \in I(j)} f(S_i)$.\; 
    Set $f(j) = c_f(j)$ and output $(j,f(j))$.\;
    Remove $f(j)$ from $f(S_i)$ for all $i \in I(j)$.\;
}
\end{algorithm}

\begin{remark}
    We briefly analyze the computational complexity of \cref{alg:peeling}. During preprocessing, computing the index sets and frequencies takes $O(mk)$ time and $O(mk+n)$ space. Since the maximum possible frequency is $m$, we can sort the coupons in $O(n)$ time using bucket sort. Inside the loop, the total time spent computing intersections and removing labels across all iterations is bounded by the sum of the sizes of the label sets. Specifically, intersecting $F(j)$ sets of size $k$ for coupon $j$ takes $O(F(j) \cdot k)$ time. Summing this cost over all $n$ coupons yields $\sum_{j=1}^n F(j) \cdot k = (mk) \cdot k = O(mk^2)$. Thus, the overall time complexity using bucket sort is $O(mk^2 + n)$ with a space complexity of $O(mk + n)$. If $m$ happens to be exceptionally large (e.g., $m > n \log n$), one might opt for a standard $O(n \log n)$ comparison-based sort instead. This is a massive improvement over a naive application of \cref{eqn:deducing}, which would require $O(n \cdot m)$ set operations taking exponential scaling in the worst case.
\end{remark}

\begin{example}
    Consider the Type-I $(8,4)$-LCCP, where we receive samples $\{(S_1,f(S_1)), (S_2,f(S_2)), (S_3, f(S_3)) \}$, with
\begin{equation*}
    S_1 = \{ 1,2,3,5\} , \quad
    S_2 = \{1, 2,4,6\}, \quad
    S_3 = \{1,3,4,7 \}.
\end{equation*}
Observe that $\{S_1,S_2, S_3\}$ is a separating system and corresponds to the draw matrix:
\begin{equation*}
    M = 
    \begin{bmatrix}
         1 & 1 & 1 & 0 & 1 & 0 & 0 & 0 \\
         1 & 1 & 0 & 1 & 0 & 1 & 0 & 0 \\
         1 & 0 & 1 & 1 & 0 & 0 & 1 & 0 
    \end{bmatrix}.
\end{equation*}

By comparing the index sets, we can obtain the Hasse diagram of the coupons with respect to $\preceq$, given by \cref{fig:hasse}.

\begin{figure}[ht!]
    \centering
    \begin{tikzpicture}[
    every node/.style={circle, draw, fill=white, inner sep=2pt, minimum size=7mm, font=\bfseries},
    every path/.style={thick}
]

    \node (1) at (0, 4.5) {1};
    \node (2) at (-2, 3) {2};
    \node (3) at (0, 3)  {3};
    \node (4) at (2, 3)  {4};
    \node (5) at (-2, 1.5) {5};
    \node (6) at (0, 1.5)  {6};
    \node (7) at (2, 1.5)  {7};
    \node (8) at (0, 0) {8};

    \draw (1) -- (2);
    \draw (1) -- (3);
    \draw (1) -- (4);

    \draw (2) -- (5);
    \draw (2) -- (6);
    \draw (3) -- (5);
    \draw (3) -- (7);
    \draw (4) -- (6);
    \draw (4) -- (7);

    \draw (5) -- (8);
    \draw (6) -- (8);
    \draw (7) -- (8);

\end{tikzpicture}
    \caption{Hasse diagram of the coupons partially ordered by $\preceq$ for Example 2.}
    \label{fig:hasse}
\end{figure}
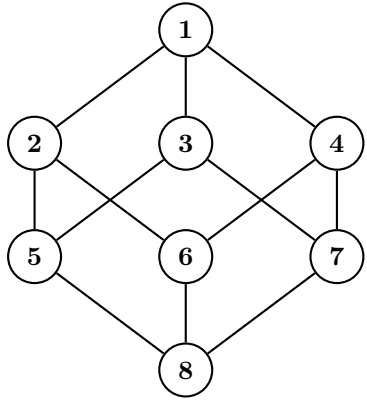

We identify coupon $1$ as the maximal element with respect to $\preceq$, so from \cref{lem:freq} we first determine $f(1)$ by performing $f(S_1) \cap f(S_2) \cap f(S_3)$. After deleting $1$ and $f(1)$ from the samples, we can repeat this process for coupons $2$, $3$, and $4$, which all appear twice in $\{S_1,S_2,S_3\}$, and then we repeat the process for coupons $5$, $6$, and $7$ which all appear only once. Finally, the label of coupon $8$ is obtained immediately once every other label has been removed from $\{f(S_1), f(S_2), f(S_3)\}$.
\end{example}

So far, we have discussed equivalent conditions and methods for determining the labels of coupons in the LCCP using the concept of separating systems. We will use separating systems in the next section to analyze the absolute minimum number of samples to deduce every coupon's label.
\section{Full Recovery Minimum} 

In this section, we find the minimum number of samples $a(n,k)$ for the $(n,k)$-LCCP in both Type-I and Type-II full recovery settings. In other words, we find \textit{minimal separating systems}, i.e., separating systems of $[n]$ consisting of the smallest possible number of $k$-subsets.

The draw matrix formulation is useful; if there exists $M \in \mathbb{F}_2^{m\times n}$ with weight-$k$ rows and distinct columns, then $a_\mathrm{I}(n,k) \leq m$. Additionally, if the columns are nonzero, then $a_{\mathrm{II}}(n,k) \leq m$. 
Therefore, a minimal separating system is given by a binary matrix with $n$ columns containing the smallest possible number of weight-$k$ rows such that all columns are distinct. Existence of such matrices will be proved via deterministic construction. 

We divide the discussion into three parts: Subsection A discusses fundamental results and properties of the minimum. Subsection B develops the full characterization of the minimum. Finally, Subsection C examines bounds on the minimum, including improvements on the bounds from~\cite{Katona,Wegener}. 

\subsection{Fundamental Results}

We begin by establishing a few logarithmic lower bounds.

\begin{outline}
    \subsubsection{Log Bound}
\end{outline}
\begin{lemma}
    \label{lemLogBound}
    For all $n, k \in \mathbb{Z}^+$ with $n > k$:
    \begin{enumerate}
        \item [(a)]\label{lemLogBound_a} $a_{\mathrm{I}}(n,k) \geq \ceil{\log_2 n}, a_{\mathrm{II}}(n,k) \geq \ceil{\log_2(n+1)}$.
        \item [(b)]\label{lemLogBound_b} $a_{\mathrm{I}}(n,k), a_{\mathrm{II}}(n,k) \geq \ceil{\log_2 2k}$.
        \item [(c)]\label{lemLogBound_c} When $n = 2k$ (resp. $n = 2k-1$) in the Type-I (resp. Type-II) setting, the lower bounds on $a_\mathrm{I}(n,k)$ (resp. $a_\mathrm{II}(n,k)$) in (a) and (b) are equal and met with equality. Furthermore, $a_{\mathrm{II}}(2k,k) = \ceil{\log_2 2k}$ if $k$ is not a power of 2.
        \item [(d)]\label{lemLogBound_d} If $k >1$, then $a_{\mathrm{I}}(2k,k) = a_{\mathrm{I}}(2k-1,k)$.
    \end{enumerate}
\end{lemma}

\begin{proof}
Suppose $\{S_i\}_{i=1}^m$ separates $[n]$ and corresponds to $M\in \mathbb{F}_2^{m\times n}$. The columns of $M$ form a subset of $\mathbb{F}_2^m$, so $n \leq |\mathbb{F}_2^m| = 2^m$, yielding the bounds in (a) (for Type-II, exclude the all-zero vector). Moreover, across all vectors in $\mathbb{F}_2^m$, each coordinate contains exactly $2^{m-1}$ ones, so $k \leq 2^{m-1}$, yielding the bound in (b).

To prove (c), we construct $M \in \mathbb{F}_2^{m\times n}$ for $m = \ceil{\log_2 2k}$ corresponding to a separating system for $[2k]$. Let the columns of $M$ consist of the all-ones vector, the all-zero vector, and $k-1$ distinct pairs of nonzero vectors that are complements of each other, which exist because $k-1 \leq 2^{m-1}-1$. This yields a separating system for $[2k]$; a separating and covering system for $[2k-1]$ is obtained by deleting the all-zero column. Furthermore, if $k$ is not a power of 2, we obtain a separating and covering system for $[2k]$ by selecting $k$ distinct pairs of nonzero vectors that are complements of each other, which exist because $k < 2^{m-1}$.

(d) holds since $\ceil{\log_2 2k} = \ceil{\log_2 (2k-1)}$ when $k > 1$.
\end{proof}

\begin{remark}
    When $k$ is a power of $2$, the draw matrix 
  for the Type-II setting constructed in  Lemma~\ref{lemLogBound}, part (c) is the parity check matrix of a Hamming code. Moreover, the respective draw matrix for the Type-I setting can be viewed as a binary splitting of the coupons where in each sample, the separation classes are split (approximately) in half. It will be seen later that the bounds in (a) become loose when $n$ is much larger than $k$.
\end{remark}

By computing the lower bound in \cref{lemLogBound_b} across all $k$ for fixed $n$, we observe that $$a_\mathrm{I}(n,k)  \geq \max \{ \ceil{\log_2 n}, \ceil{\log_2 2k} \},$$ and $$a_\mathrm{II}(n,k)  \geq \max \{ \ceil{\log_2(n+1)},   \ceil{\log_2 2k} \},$$ which gives the following result for when the subsets can be of any cardinality.

\begin{corollary}
    \label{corRandMin}
    Suppose we have $n$ coupons and can sample subsets of any size. Then, the minimum number of samples to deduce the label of each coupon is $\ceil{\log_2 n}$ in the Type-I setting and $\ceil{\log_2 (n+1)}$ in the Type-II setting.
\end{corollary}
\begin{proof}
    Clearly, these are lower bounds since the bounds in  \cref{lemLogBound_a}(a) do not depend on the size of the sampled subsets. Moreover, these bounds are tight since we can simply pick any $n$ distinct (resp. and nonzero) columns to construct a separating (resp. and covering) system.    
\end{proof}

From \cref{thm_deducing}, it is clear that $a_{\mathrm{II}}(n,k) \geq a_{\mathrm{I}}(n,k)$, though we can characterize the relationship further.

\begin{outline}
    \subsubsection{Relationship between Type-I,II min}
\end{outline}

\begin{lemma}
    \label{lemTypeMin}
    For any $k\in \mathbb{Z}^+$:
    \begin{equation}
        a_\mathrm{II}(n,k) = \begin{cases}
            a_\mathrm{I}(n,k) , & k < n < 2k \\
            a_\mathrm{I}(n+1,k), & n \geq 2k-1.
        \end{cases}
    \end{equation}
    Moreover, $a_\mathrm{II}(n,k) - a_{\mathrm{I}}(n,k) \in \{0, 1\}$ when $n \geq 2k-1$. Consequently, $a_\mathrm{I}(n,k)$ and $a_\mathrm{II}(n,k)$ are $1$-Lipschitz functions with respect to $n$ for fixed $k$.
\end{lemma}

\begin{proof}
    Suppose $n < 2k$. Let $m = a_{\mathrm{I}}(n,k)$ and $\{S_i\}_{i=1}^m$ be a separating system corresponding to $M \in \mathbb{F}_2^{m\times n}$, so each row of $M$ has weight $k$, and the columns are distinct. If $\{S_i\}_{i=1}^m$ covers $[n]$, then we can conclude that $a_{\mathrm{I}}(n,k) = a_{\mathrm{II}}(n,k)$. If not, assume, without loss of generality, that the last column of $M$ is the all-zero vector. By \cref{lemLogBound_b}(b), $2^m \geq 2k > n$, so there is a nonzero $v \in \mathbb{F}_2^m$ that is not a column of $M$. 
    
    We proceed iteratively to transform the last column into $v$. For each row index $i$ where $v_i = 1$, we transfer a $1$ from one of the first $n-1$ columns to the last column. Since this transfer only occurs in rows where $v_i = 1$, it is impossible for any of the first $n-1$ columns to become $v$. So it suffices to ensure that for all $i$ where $v_i = 1$, we can select at least one column whose $1$ in row $i$ can be transferred while keeping the first $n-1$ columns distinct. 
    
    Consider any such row $i$. The only way to create a duplicate column is to transfer a $1$ from a column that, after the flip, becomes identical to another existing column. Excluding the last column, let $N_1$ be the number of columns with a $1$ in row $i$ and $N_0$ be the number of columns with a $0$. Since the row weight is $k$ and the last column is now $0$, we have $N_1 = k$ and $N_0 = (n-1) - k$. The difference is:
    \[ N_1 - N_0 = k - (n - k - 1) = 2k - n + 1. \]
    Since $n < 2k$, we have $N_1 - N_0 \geq 2$, meaning $N_1 > N_0$. Thus, there are strictly more potential donors (columns with $1$) than potential conflicts (columns with $0$), ensuring there is always at least one valid choice of column for the transfer.
    
    Next, suppose $n \geq 2k-1$. It is clear that $a_\mathrm{I}(n+1,k) \leq a_{\mathrm{II}}(n,k)$, since a separating and covering system for $[n]$ is a separating system for $[n+1]$ (by adding an all-zero column). To prove the converse, start with a separating system for $[n+1]$ corresponding to $M \in \mathbb{F}_2^{m\times (n+1)}$. If $M$ has an all-zero column, we can delete it to obtain the result. Otherwise, assume $M$ has no all-zero column. The strategy is to transfer every $1$ from the last column to a column that currently has a $0$, while keeping columns distinct, and then delete the resulting all-zero last column.
    
    Consider any row where the last column has a $1$. The only way to create a duplicate column is to transfer the $1$ to a column that, after the flip, becomes identical to another existing column. Excluding the last column, let $N_1$ be the number of columns with a $1$ and $N_0$ be the number of columns with a $0$. Since the row weight is $k$ and the last column is $1$, we have $N_1 = k-1$ and $N_0 = n - (k-1)$. The difference is:
    \[ N_0 - N_1 = (n - k + 1) - (k - 1) = n - 2k + 2. \]
    Since $n \geq 2k-1$, we have $N_0 - N_1 \geq 1$, meaning $N_0 > N_1$. Thus, there are strictly more potential receivers (columns with $0$) than potential conflicts (columns with $1$), ensuring there is always at least one valid choice of column to receive the $1$.
    
    Finally, suppose $a_\mathrm{II}(n,k) > a_\mathrm{I}(n,k)$ for $n \geq 2k-1$. Let $m = a_{\mathrm{I}}(n,k)$ and $\{S_i \}_{i=1}^m $ be a separating system corresponding to draw matrix $M\in \mathbb{F}_2^{m\times n}$. By the hypothesis, $\{S_i\}_{i=1}^m$ does not cover $[n]$, so there will be exactly one all-zero column in $M$, denoted as column $j$. Then, define $M'$ by inserting in $M$ a weight-$k$ row with a $1$ in the $j$-th column. Now, all columns of $M'$ are nonzero and distinct, so $a_{\mathrm{II}}(n,k) = a_{\mathrm{I}}(n,k) + 1$.
\end{proof}

\begin{example}
\label{ex:transfer_min}
    Consider the case $n=5$ and $k=3$, where $a_\mathrm{I}(5,3)$ and $a_\mathrm{II}(5,3)$ are at least $3$ by \cref{lemLogBound_a}. Let $M$ be the following draw matrix:
    \begin{equation*}
        M = 
        \begin{bmatrix}
        1 & 1 & 1 & 0 & 0 \\
        1 & 1 & 0 & 1 & 0 \\
        1 & 0 & 1 & 1 & 0 
        \end{bmatrix}.
    \end{equation*}
    Note that $M$ corresponds to a separating system but not a covering system. To remedy this, we must transfer ones to the fifth column to change it into a nonzero column vector that is not identical to an existing column of $M$. There are three options: $[1,0,0]^T$, $[0,1,0]^T$, and $[0,0,1]^T$. Here, we pick $[1,0,0]^T$, which means the transfer occurs in the first row of $M$.
    The donor candidates in row $1$ are columns $1$, $2$, and $3$, but note that column $1$ matches with column $4$ in both rows $2$ and $3$, so the only valid choices of columns to transfer from are columns $2$ and $3$. We transfer the $1$ from column $2$, so the new column $2$ is $[0,1,0]^T$ and the new column $5$ is $[1,0,0]^T$. Then, the modified draw matrix is:
    \begin{equation*}
        M' = 
        \begin{bmatrix}
        1 & 0 & 1 & 0 & 1 \\
        1 & 1 & 0 & 1 & 0 \\
        1 & 0 & 1 & 1 & 0 
        \end{bmatrix}.
    \end{equation*}

    Observe that now $M'$ is a valid minimal Type-II draw matrix obtained from $M$.
\end{example}

An immediate consequence of \cref{lemTypeMin} is the following monotonicity  property.

\begin{outline}
    \subsubsection{Monotonicity}
\end{outline}
\begin{corollary}
\label{corMinMonotonicity} When $n \geq 2k-1$, 
$$a_{\mathrm{I}}(n+1,k) =a_{\text{II}}(n,k) \geq a_{\mathrm{I}}(n,k),$$ $$a_{\mathrm{II}}(n+1,k) = a_\mathrm{I}(n+2,k) \geq a_{\text{I}}(n+1,k) \geq a_{\mathrm{II}}(n,k).$$
\end{corollary}

We also have the following symmetry property for the Type-I minimum which was also proved in~\cite{Katona}.

\begin{outline}
    \subsubsection{Symmetry}
\end{outline}
\begin{lemma}
\label{lemMinSymmetry}
     $a_{\mathrm{I}}(n,k) = a_{\mathrm{I}}(n,n-k)$.
\end{lemma}
\begin{proof}
    Take any matrix $M \in \mathbb{F}_2^{a_\mathrm{I}(n,k) \times n}$ with weight-$k$ rows and distinct columns; its component-wise complement $M^\mathrm{c}$ will have weight-$(n-k)$ rows and distinct columns, so $a_{\mathrm{I}}(n,n-k) \leq a_{\mathrm{I}}(n,k)$. Applying the same argument in the other direction yields the result.
\end{proof}
\begin{remark}
    It follows from \cref{lemTypeMin,lemMinSymmetry} that the scope of computing $a_\mathrm{I}(n,k)$ and $a_\mathrm{II}(n,k)$ can be reduced to the regime $n \geq 2k$, since if $n < 2k$, then $n > 2(n-k)$.  
\end{remark}

Next, we can upgrade \cref{corMinMonotonicity} by showing a full ordering of $a_\mathrm{I}(n,k)$ and $a_{\text{II}}(n,k)$ with respect to $n$ and $k$ for when one of the parameters is fixed. This is done in \cref{lemMinOrderingII,lemMinOrderingI}.

\begin{lemma} For $n,k\in \mathbb{Z}^+$ and $n>k$:
   \label{lemMinOrderingII}
   \begin{enumerate}
       \item[(a)] If $n \geq 2k-1$, then $a_{\mathrm{II}}(n+1,k) \geq a_{\mathrm{II}}(n,k) $.
       \item[(b)] If $n \geq 2k + 1$ (or $k \leq \ceil{\frac{n}{2}} - 1$), then $a_{\mathrm{II}}(n,k+1) \leq a_{\mathrm{II}}  (n,k)$.
       \item[(c)] If $k < n < 2k-1$, then $a_{\mathrm{II}}(n+1,k) \leq a_{\mathrm{II}} (n,k) $.
       \item[(d)] If $k < n < 2k+1$ (or $k \geq \ceil{\frac{n}{2}}$), then $a_{\mathrm{II}}(n,k+1) \geq a_{\mathrm{II}}(n,k)$.
   \end{enumerate}
\end{lemma}

\begin{proof}
    We prove each part by starting with a minimal draw matrix $M$ and modifying its rows. In each row, we must select a candidate column to modify without creating a duplicate. A conflict occurs only if modifying a candidate makes it identical to an existing column that differs only in that row. By comparing the number of available candidates ($N_{cand}$) and potential conflicts ($N_{conf}$), the Pigeonhole Principle guarantees a safe choice whenever $N_{cand} - N_{conf} \ge 1$. To additionally ensure we do not create the forbidden all-zero column, we require $N_{cand} - N_{conf} \ge 2$.
    \begin{enumerate}
        \item[(a)] This was given by  \cref{corMinMonotonicity}.
        
        \item[(b)] We flip a $0$ to a $1$ in each row of $M \in \mathbb{F}_2^{m \times n}$ (which has weight $k$). Here, $N_{cand} = n-k$ (number of $0$'s) and $N_{conf} = k$ (number of $1$'s). Since $n \ge 2k+1$, $N_{cand} - N_{conf} = n-2k \ge 1$. This ensures a safe $0$ to flip (and adding $1$'s cannot create an all-zero column). Thus, $a_{\mathrm{II}}(n, k+1) \le a_{\mathrm{II}}(n, k)$.
        
        \item[(c)] We append a column of zeros to $M \in \mathbb{F}_2^{m \times n}$ and transfer a $1$ to it in rows where a new non-zero target vector $v$ (which exists since $n < 2k-1 \le 2^m-1$) has a $1$. Here, $N_{cand} = k$ and $N_{conf} = n-k$. Since $n < 2k-1$, $N_{cand} - N_{conf} = 2k-n \ge 2$. This guarantees a safe $1$ to transfer while avoiding the all-zero column. Thus, $a_{\mathrm{II}}(n+1, k) \le a_{\mathrm{II}}(n, k)$.
        
        \item[(d)] We flip a $1$ to a $0$ in each row of $M \in \mathbb{F}_2^{m \times n}$ (which has weight $k+1$). Here, $N_{cand} = k+1$ and $N_{conf} = n-(k+1)$. Since $n < 2k+1$, $N_{cand} - N_{conf} = 2k-n+2 \ge 2$. This guarantees a safe $1$ to flip without creating an all-zero column. Thus, $a_{\mathrm{II}}(n, k) \le a_{\mathrm{II}}(n, k+1)$.
    \end{enumerate}
\end{proof}

\begin{lemma}
    \label{lemMinOrderingI}
    The inequalities in \cref{lemMinOrderingII} also hold for the Type-I minimum.
\end{lemma}

\begin{proof}
    The results follow by utilizing the same modification strategies established in the proof of \cref{lemMinOrderingII}. Because the Type-I setting permits the all-zero column, distinctness is the only constraint, which requires $N_{cand} - N_{conf} \ge 1$. This condition is strictly satisfied across all four cases. 
\end{proof}
\begin{corollary}
    \label{corMinOptimization}
    For positive integers $n,k$ with $n > k$:
    \begin{align*}
        \min\limits_{n} a_\mathrm{I}(n,k) &= a_\mathrm{I} (2k-1,k) = \ceil{\log_2 2k}, \\ 
        \min\limits_{n} a_\mathrm{II}(n,k) &= a_\mathrm{II}(2k-1,k) = \ceil{\log_2 2k},\\
        \min\limits_k a_\mathrm{I}(n,k) &= a_\mathrm{I}\left(n, \ceil{\frac{n}{2}}\right) = \ceil{\log_2 n},\\
        \min\limits_k a_\mathrm{II}(n,k) &= a_\mathrm{II}\left(n, \ceil{\frac{n}{2}}\right) = \ceil{\log_2 (n+1)}.
    \end{align*}

\end{corollary}
\begin{proof}
    Immediate from \cref{lemLogBound_b,lemMinOrderingII,lemMinOrderingI}.
\end{proof}

\begin{outline}
    \subsubsection{Subadditivity and Subhomogeneity}
\end{outline}

\begin{lemma}
    \label{lemma_subhom}
    For $n_1,n_2,k,p \in \mathbb{Z}^+$ such that $n_1,n_2\geq 2k -1$:
    \begin{enumerate}
    \item [(1)] $a_{\mathrm{II}}(n_1+n_2, k) \leq a_\mathrm{II}(n_1,k) + a_\mathrm{II}(n_2,k)$, and 
    \item [(2)] $a_{\mathrm{II}}(pn_1, k) \leq p \; a_{\mathrm{II}}(n_1,k)$.
    \end{enumerate}
    
\end{lemma}
\begin{proof}
    Let $m_i = a_\mathrm{II}(n_i,k)$ and $M_i \in \mathbb{F}_2^{m_i \times n_i}$ correspond to a separating and covering system for $[n_i]$ for $i \in \{1,2\}$, and define $M\in \mathbb{F}_2^{(m_1+m_2)\times (n_1 + n_2)}$ as:
    \begin{equation*}
        M = \begin{bmatrix}
            M_1 & 0 \\ 0 & M_2
        \end{bmatrix}.
    \end{equation*}
    Then, observe that $M$ will have weight-$k$ rows and distinct nonzero columns (since neither $M_1$ nor $M_2$ contains all-zero columns), so subadditivity in (1) holds. Subhomogeneity in (2) is obtained via repeated applications of subadditivity.
\end{proof}
\begin{remark}
    From \cref{lemma_subhom}, one can immediately derive upper bounds on the minimum when $k$ is fixed. In particular, this provides a simple method for constructing a separating system by combining multiple smaller separating systems, with the tradeoff being that the combined separating system is not necessarily minimal. 
\end{remark}

\subsection{Full Characterization}
We begin with a lower bound that is useful when $n$ is much larger than $k$.

\begin{outline}
    \subsubsection{Efficiency Bound}
\end{outline}
\begin{lemma} 
\label{lemma_fixedk}
$a_\mathrm{I}(n,k) \geq \ceil{\frac{2n-2}{k+1}}$ and $a_{\mathrm{II}}(n,k) \geq \ceil{\frac{2n}{k+1}}$.
\end{lemma}

\begin{proof}
    Consider any separating system $\{S_i\}_{i=1}^m$ for some $m \in \mathbb{Z}^+$ corresponding to draw matrix $M \in \mathbb{F}_2^{m \times n}$. For $w\in \{0,\dots,m\}$, let $A_w$ be the number of weight-$w$ columns, so $\sum\limits_{w=0}^m A_w = n$ and $A_w \leq \binom{m}{w}$ (since columns are distinct). Then, the total number of ones in $M$ is $$mk = \sum\limits_{w=0}^m A_w w = A_1 + \sum\limits_{w=2}^m A_w w \geq A_1 + 2(n-A_0-A_1),$$ implying  $2n \leq mk + 2A_0 + A_1$. 

   The bound for Type-I follows from $A_0 \leq 1, A_1 \leq m$, and for Type-II from  $A_0 = 0, A_1 \leq m$.
\end{proof}
\begin{remark}
    The bound in Lemma~\ref{lemma_fixedk}  is tight when $n \geq \binom{k+1}{2}$ (see \cref{lemMinLargen}), whereas the bound in \cref{lemLogBound} is loose in this regime, as  illustrated by \cref{ex: opt min}. 
\end{remark}

\begin{example}\label{ex: opt min}
    Consider the Type-II $(15,5)$-LCCP and subsets $\{S_i\}_{i=1}^5 \subseteq [15]$ shown in \autoref{fig:const_min}, where row $i$ corresponds to set $S_i$ and column $j$ corresponds to coupon $j$. It can be seen that $15 = \binom{5+1}{2}$, and that $\{S_i\}_{i=1}^5$ separates and covers $[15]$, achieving the lower bound  $\frac{2n}{k+1} = 5$. 
\end{example}
\begin{figure}[ht]
   \centering
\resizebox{0.8\linewidth}{!}{%
    \renewcommand{\arraystretch}{1.5}
    \begin{NiceTabular}{*{15}{C}}[first-row, first-col, hvlines]
    
    & 1 & 2 & 3 & 4 & 5 & 6 & 7 & 8 & 9 & 10 & 11 & 12 & 13 & 14 & 15 \\ 

    $S_1$ & \clg{1} & \clg{1} & \clg{1} & \clg{1} & \clg{1} & & & & & & & & & & \\ 
    $S_2$ & \clg{1} & & & & & \clg{1} & \clg{1} & \clg{1} & \clg{1} & & & & & & \\ 
    $S_3$ & & \clg{1} & & & & \clg{1} & & & & \clg{1} & \clg{1} & \clg{1} & & & \\ 
    $S_4$ & & & \clg{1} & & & & \clg{1} & & & \clg{1} & & & \clg{1} & \clg{1} & \\ 
    $S_5$ & & & & \clg{1} & & & & \clg{1} & & & \clg{1} & & \clg{1} & & \clg{1} \\ 

    \end{NiceTabular}%
}
    \caption{Visualization of a separating and covering system achieving the minimum of 5 samples for the Type-II $(15,5)$-LCCP. } 
    \label{fig:const_min}
\end{figure}

The proof of \cref{lemma_fixedk} suggests that it may be useful to consider the dual problem: namely, given a fixed number of samples $m$ and drawing size $k$, what is the \textit{maximum} number of coupons whose labels can be determined in $m$ samples? This can be achieved by using low column weights in the draw matrix, such as in \cref{ex: opt min} where each column has a weight of $1$ or $2$. The analysis of such matrices is facilitated by introducing the following definitions.

\begin{definition}
   Denote the volume of a Hamming ball in $\mathbb{F}_2^m$ with radius $r$ as $\Vol_2(r,m) = \sum\limits_{\ell=0}^r \binom{m}{\ell}$. For $n \in \mathbb{Z}^+$ such that $n \leq 2^m$, define its \textit{minimum Hamming radius} $r_m(n)$ as the smallest value of $r$ such that $\Vol_2(r,m) \geq n$. In other words, 
   \begin{equation}
        \label{eqn:HammingRadius}
       r_m(n) = \min\{ r : \Vol_2(r,m) \geq n\}.
   \end{equation}
\end{definition}

\begin{definition}
    Given $m,w,k,n \in \mathbb{Z}^+$ such that $m \geq w, n \geq k$, we say that a binary matrix $M\in \mathbb{F}_2^{m\times n}$ is $(w,k)$-\textbf{almost regular} if every column has weight $w$ and every row has weight either $k$ or $k-1$.  
\end{definition}

\begin{remark}
    By dividing both sides of the inequality condition in \cref{eqn:HammingRadius} by $2^m$, we obtain the inequality $\frac{\Vol_2(r,m)}{2^m} \geq \frac{n}{2^m}$, where $\frac{n}{2^m}\in (0,1)$. The left-hand side can be written as $\Pr(X \leq r) \geq \frac{n}{2^m}$, where $X$ is binomially distributed with parameters $m$ and $1/2$. Thus, $r_m(n)$ is the inverse binomial cdf of $\frac{n}{2^m}$ with respect to parameters $m$ and $1/2$.
\end{remark} 

\begin{outline}
    \subsubsection{Almost Regular Construction}
\end{outline}

A useful result is that for certain values of $m,w,k,n$, we can guarantee the existence of an almost-regular matrix via construction.

\begin{proposition}
\label{prop_almost}
    Given $m,w,k,n\in \mathbb{Z}^+$ such that $m \geq w, k \leq \binom{m-1}{w-1}, n = \floor{\frac{mk}{w}}$, then there exists a $(w,k)$-almost regular matrix $M \in \mathbb{F}_2^{m\times n}$ such that the columns of $M$ are distinct, and more than $m-w$ rows have weight $k$.
\end{proposition}
\begin{proof}
    See Appendix.
\end{proof}

\begin{outline}
    \subsubsection{Thresholds on Max Coupons with Fixed Number of Samples}
\end{outline}
\begin{lemma}
\label{lemThresholds}
For a fixed $k \in \mathbb{Z}^+$, let $m$ be any integer such that $m \geq \ceil{\log_2 2k}$. Define:
\begin{equation}
    \hat{k}_m  = \Vol_2 (r_{m-1}(k) - 1, m-1),
\end{equation}
\begin{equation}
\label{eqn:n_m}
    n_m = \Vol_2 (r_{m-1}(k), m)- 1 + \floor{\frac{m(k-\hat{k}_m)}{r_{m-1}(k) + 1}}.
\end{equation}
Then the sequence $\{n_m\}$ is strictly increasing with $m$. Furthermore, $n_m \geq 2k-1$ and $a_\mathrm{I}(n_m + 1,k) = a_\mathrm{II}(n_m, k) = m$.
\end{lemma} 

\begin{proof}
    We prove the result for $a_{\mathrm{II}}(n_m,k)$. The result for $a_{\mathrm{I}}(n_m+1,k)$ will follow immediately from \cref{lemTypeMin}, since it will be shown that $n_m \geq 2k-1$. 

    First, we present the following necessary condition for the existence of a separating and covering system. With sampling of subsets of size $k$, we claim there exists $M \in \mathbb{F}_2^{m\times n}$ corresponding to a separating and covering system for $[n]$ only if there exist integers $A_1,\dots,A_m \in \mathbb{Z}^+$ such that  
    \begin{align}
    \label{eqn_valid}
         n& = \sum\limits_{w=1}^m A_w\nonumber,\\
        mk &= \sum\limits_{w=1}^m w A_w \nonumber,\\ 
        A_w &\leq \binom{m}{w} \; \forall \; w=1,\dots,m.
    \end{align}
    In other words, to prove $a_{\mathrm{II}}(n_m, k) = m$, we must show that a valid draw matrix for $n_m$ coupons exists when the number of rows is $m$, but no such matrix exists when the number of rows is $m-1$. Thus, we evaluate the system of equations \eqref{eqn_valid} specifically for $n = n_m$. We will demonstrate this after establishing that a solution exists for $n = n_m$ with $m$ rows, and providing an explicit construction for a separating and covering system.

    Set $r = r_{m-1}(k)$, and consider the following maximization problem:
    \begin{equation}
    \label{eqn_max}
    \max\limits_{\substack{ \sum\limits_{w=1}^m wA_w \leq mk \\A_w \leq \binom{m}{w} \; \forall w = 1,\dots m}} \; \sum\limits_{w=1}^m A_w. 
    \end{equation}
    
    Expressed in terms of the well-known knapsack problem, the goal here is to maximize the number of items carried, given a total weight limit of $mk$ and a supply of $\binom{m}{w}$ items of weight $w$, for $w = 1,\dots,m$. Note that the maximum of~\eqref{eqn_max} is an upper bound to the maximum value of $n$ where a solution to $\eqref{eqn_valid}$ exists, though we later observe that these two quantities are equal. Observe that one maximizer for \eqref{eqn_max} is:
    \begin{equation}
    \label{eqn_maximizer}
        A_w = \begin{cases}
            \binom{m}{w}, & 1 \leq w \leq r  \\
             \floor{\frac{m(k- \hat{k}_m)}{r+1}}, & w = r+1\\
             0, & \mathrm{otherwise}.
        \end{cases}
    \end{equation}
    
    Indeed, since $w \binom{m}{w} = m\binom{m-1}{w-1}$, then 
    \begin{align*}
        \sum\limits_{w=1}^m wA_w = \sum\limits_{w=1}^r w \binom{m}{w} + (r+1)A_{r+1} \\
        = \sum\limits_{w=1}^r m\binom{m-1}{w-1} + (r+1)A_{r+1} = m \hat{k}_m + (r+1)A_{r+1}.
    \end{align*}
    The value $d = mk - \sum_{w=1}^m wA_w$ is exactly the remainder of dividing the leftover capacity $m(k-\hat{k}_m)$ by $(r+1)$, meaning $d \le r$. No more items can be added because the remaining capacity is $d \le r$, and the supply of all items with weight $\le r$ is already fully exhausted. Furthermore, no item can be replaced by two other items to increase the total count: removing our heaviest item ($r+1$) and adding our maximum leftover budget ($r$) yields a maximum available capacity of $2r+1$. Because the lightest available items have weight $r+1$, acquiring two new items requires a capacity of at least $2r+2$. Thus, a 1-for-2 swap is strictly impossible. Moreover, the total number of items is $\sum\limits_{w=1}^m A_w = \sum\limits_{w=1}^r \binom{m}{w} + A_{r+1} = \Vol_2 (r,m) -1 + A_{r+1} = n_m$.
    
    The $A_w$ values in \eqref{eqn_maximizer} do not necessarily satisfy $mk = \sum\limits_{w=1}^m wA_w$. However, if $d > 0$, this can be fixed by decreasing $A_{r+2-d}$ by $1$, and increasing $A_{r+2}$ by $1$. This is valid because our greedy assignment set $A_{r+2} = 0$, allowing us to safely increase it without violating the upper bound. Furthermore, $A_{r+2-d}$ corresponds either to a fully populated weight class (if $d \ge 2$) or a populated one (if $d=1$, where $A_{r+1} \ge 1$ because the leftover capacity $m(k-\hat{k}_m)$ is a multiple of $m \ge 2$ and cannot be exactly $1$). This guarantees $A_{r+2-d} \ge 1$, ensuring we never decrease a variable below zero. Thus, this modification provides a solution to \eqref{eqn_valid} when $n = n_m$ and the number of rows is $m$.

    Now, we construct $M \in \mathbb{F}_2^{m \times n}$ with $n= n_m$ distinct, nonzero columns and weight-$k$ rows. Note that if $r = m-1$, then we would have $k > \Vol_2(m-2,m-1) = 2^{m-1}-1$ and $k \leq \Vol_2(m-1,m-1) = 2^{m-1} \Rightarrow k = 2^{m-1}$, and so the columns of $M$ will consist of all distinct nonzero vectors in $\mathbb{F}_2^m$. Thus, we assume that $r \leq m-2$.
     
    Let $V_w \subseteq \mathbb{F}_2^m$ be the set of weight-$w$ tuples for $w \in \{1,\dots,r\}$. The total number of $1$'s in each row in $V_w$ is $\binom{m-1}{w-1}$, and $|V_w| = \binom{m}{w}$. Inserting the tuples in $V_1,\dots,V_r$ as columns of $M$ tentatively gives it $\hat{k}_m$ $1$'s in each row and $\Vol_2(r,m)-1$ columns. We then populate the rest of $M$ using weight-$(r+1)$ columns. Since $k - \hat{k}_m \leq \Vol_2(r,m-1) - \Vol_2(r-1,m-1) = \binom{m-1}{r}$, then by \cref{prop_almost}, there exists a $(r+1,k - \hat{k}_m)$-almost regular matrix with $\floor{\frac{m(k-\hat{k}_m)}{r+1}}$ distinct columns and more than $m-r-1$ rows with weight $k-\hat{k}_m$. Inserting this as a submatrix of $M$ will give us the correct number of columns, but there may still be rows with weight $k-1$.
    
    Namely, there are $d$ such ``deficient'' rows with weight $k-1$, where $d \leq r$. If $d > 0$, then we resolve the deficiency by finding a column of weight $r+2-d$ with zeros in every deficient row. To see why such a column exists, note that when $d \ge 2$, $M$ contains the complete set of weight-$(r+2-d)$ vectors. Since $m - d \ge r+2-d$, there exists at least one vector with ones exclusively outside the $d$ deficient rows. Alternatively, if $d=1$, we need a weight-$(r+1)$ column with a zero in the single deficient row. Because $m \ge r+2$, the number of columns in the almost-regular submatrix satisfies $A_{r+1} \ge k-\hat{k}_m$. Since the deficient row contains exactly $k-\hat{k}_m-1$ ones within this submatrix, there is at least $A_{r+1} - (k-\hat{k}_m-1) \ge 1$ column with a zero in that row. We simply isolate this column and flip its $d$ zeros to ones, transforming it into a valid weight-$(r+2)$ column. Hence, this gives a separating and covering system for $[n_m]$, so $a_{\mathrm{II}}(n_m,k) \leq m$.

    Additionally, for all $m \geq \ceil{\log_2 2k}$, $n_m$ is the maximum of \eqref{eqn_max}, and more importantly is the maximum $n$ for which $[n]$ can be separated and covered by $m$ samples. Since $m \geq \ceil{\log_2 2k}$, $[2k-1]$ can be separated and covered by $m$ samples (via \cref{lemLogBound} part (c)), so $n_m \geq 2k-1$ by maximality.

    Finally, to show that $a_{\mathrm{II}}(n_m, k) \geq m $, we show that no separating and covering system of $[n_m]$ exists for $m-1$ samples, namely, by showing that $n_m$ cannot be achieved in \eqref{eqn_max} with $m$ replaced by $m-1$ without violating the total weight constraint. Moreover, since $n_{m-1}$ is the maximum $n$ for which $[n]$ can be separated and covered by $m-1$ samples, this will also imply that $n_{m-1} < n_m$, or in other words, $\{n_m\}$ is strictly increasing (note that if $m = \ceil{\log_2 2k}$, then $m-1$ falls below the absolute minimum required samples established in \cref{lemLogBound}, meaning no valid system can exist, so $n_{m-1}$ is strictly less than $2k-1 \le n_m$). 
    
    We start with the earlier solution from \eqref{eqn_maximizer}, which now violates the supply constraint for each $w \leq r$, and remove the ``overfill'' in each weight class. The total number of items removed is $\sum\limits_{w=1}^r \binom{m}{w}-\binom{m-1}{w} = \sum\limits_{w=1}^r \binom{m-1}{w-1} = \hat{k}_m$, and the total weight is decreased by $\sum\limits_{w=1}^r w \binom{m-1}{w-1}$. Since the total weight limit has been reduced by $k$, then we must add $\hat{k}_m$ items of weight at least $r+1$ so that their total weight is at most $d + \sum\limits_{w=1}^r w \binom{m-1}{w-1} -k$, where $d\leq r$. However, since $k > \hat{k}_m$ and $\hat{k}_m = \sum\limits_{w=0}^{r-1} \binom{m-1}{w} \geq r$, then $d + \sum\limits_{w=1}^r w \binom{m-1}{w-1} -k< r + r\hat{k}_m -\hat{k}_m = r+(r-1)\hat{k}_m \leq r\hat{k}_m <(r+1) \hat{k}_m$, which means that the total weight limit cannot be satisfied with $n_m$ items. Therefore, $a_{\mathrm{II}}(n_m,k) = m$.
\end{proof}

\begin{remark}
    \label{rmkThresholds}
    Note that $n_m$ can be viewed as a threshold determining the maximum number of coupons whose labels can be deduced with $m$ samples. Moreover, we note that although~\cite{Katona} proved the existence of a separating system if and only if a system of equations similar to \eqref{eqn_valid} can be satisfied, our proof directly provides a method for constructing such a system using the proofs of \cref{lemThresholds} and \cref{prop_almost}. We provide the full construction in \cref{alg:maximal_lccp} in the Appendix.

    Moreover, when $m \geq k$, observe that $r_{m-1}(k) = 1$, so $M$ will only have weight-$1$ or weight-$2$ columns (with a possible exception of one column), which will be useful for \cref{lemMinLargen} below. Thus, \cref{lemThresholds} can be viewed as a generalization of \cref{lemma_fixedk} by more accurately counting the number of $1$'s in the draw matrix.

    Finally, from this proof it is also apparent that if the setup of the LCCP were to be modified such that each subset has \textit{at most} $k$ coupons, then the minimum number of samples would still be the same. This can be seen by obtaining a minimal separating system before inserting $1$'s in the matrix to ensure every row has weight-$k$. This phenomenon was also observed in~\cite{Katona}, and since we will use it to prove another result we state it here.
\end{remark}

\begin{lemma}
    \label{lemAtMostEquivalence}
    For positive integers $n,k$ such that $n \geq 2k$, $a(n,k)$ is also equal to the minimum number of samples of subsets of size \textit{at most} $k$ needed for a separating system.
\end{lemma}
\begin{proof}
    Observation from \cref{rmkThresholds}.
\end{proof}
 
The following result is a consequence of \cref{lemTypeMin,lemMinOrderingII,lemMinOrderingI,lemma_fixedk,lemThresholds}. It was first proved by other means in~\cite{Katona}.

\begin{outline}
    \subsubsection{Minimum for Large n}
\end{outline}

\begin{lemma}
 \label{lemMinLargen} Lemma~\ref{lemma_fixedk} holds with equality if $n \geq \binom{k+1}{2}$.
\end{lemma}

\begin{proof}
Consider the Type-II bound. Setting $m = \ceil{\frac{2n}{k+1}}$, observe that $\Vol_2(1,m-1) =  m \geq \frac{2\binom{k+1}{2}}{k+1} = k$, so $r_{m-1}(k) = 1$. By \cref{lemThresholds}, $a_\mathrm{II}(n', k) = m$, where $n' = \Vol_2(1,m)-1 + \floor{\frac{m(k-\Vol_2(0,m-1))}{1+1}} =  m + \floor{\frac{m(k-1)}{2}}$, which can also be written as $n' = m-1 + \ceil{\frac{m(k-1)+1}{2}} = \ceil{\frac{m(k+1)-1}{2}}$. Moreover, $m = \ceil{\frac{2n}{k+1}} \geq \frac{2n}{k+1}$, so $n' \geq \ceil{\frac{\frac{2n}{k+1}(k+1) -1}{2}} = \ceil{n-\frac{1}{2}}=n$. By \cref{lemMinOrderingII}, $a_\mathrm{II}(n,k) \leq  m$. 

Finally, by \cref{lemma_fixedk}, $a_{\mathrm{II}}(n,k) = m$.  The respective expression for Type-I follows immediately from \cref{lemTypeMin} since $n \geq \binom{k+1}{2}\geq2k-1$.
\end{proof}

\begin{corollary}
    $a_{\mathrm{I}}(\binom{k+1}{2}p+1,k) = a_{\mathrm{II}}(\binom{k+1}{2}p,k) = kp$ for any $p \in \mathbb{Z}^+$. 
\end{corollary}

\begin{corollary}
    $a_{\mathrm{I}}(n,2)= \ceil{\frac{2n-2}{3}}$ and $ a_{\mathrm{II}}(n,2) = \ceil{\frac{2n}{3}}$.
\end{corollary}

\cref{thm_minimum_full} fully characterizes $a_{\mathrm{I}}(n,k)$ and $a_{\mathrm{II}}(n,k)$. 

\begin{outline}
    \subsubsection{Minimum Theorem}
\end{outline}
\begin{theorem}
    \label{thm_minimum_full}
    Given any $k \in \mathbb{Z}^+$, define the sequence $\{n_m\}$ for $m \geq \ceil{\log_2 2k}$ as in \cref{lemThresholds}.
    \begin{enumerate}
        \item[(a)] If $n \geq \binom{k+1}{2}$, then $a_{\mathrm{I}}(n+1,k) = a_\mathrm{II} (n,k) = \ceil{\frac{2n}{k+1}}$.
        \item[(b)] If $2k-1 \leq n_{m-1} < n\leq n_m \leq \binom{k+1}{2}$ for some $m$, then $a_\mathrm{I}(n+1,k) =a_{\mathrm{II}} (n,k) =   m$.
        \item [(c)] If $k < n<2k$, then $a_\mathrm{I}(n,k) = a_\mathrm{II}(n,k) = a_\mathrm{I}(n,n-k)$.
    \end{enumerate}
\end{theorem}

\begin{proof}
    (a) follows from \cref{lemMinLargen}, and (c) follows from \cref{lemTypeMin,lemMinSymmetry}. If $n$ falls in (b), then by the proof of \cref{lemThresholds}, $a_\mathrm{II}(n,k) > m-1$, and by \cref{lemMinOrderingII},  $a_\mathrm{II}(n,k) \leq a_{\mathrm{II}}(n_m,k) = m$, so $a_\mathrm{I}(n+1,k) =a_{\mathrm{II}} (n,k) = m$.
\end{proof}

\begin{remark}
    \label{rmk:minimum_full}
    \cref{thm_minimum_full} provides a procedure for finding the minimum values $a_\mathrm{I}(n,k)$ and $a_\mathrm{II}(n,k)$. 
    \begin{enumerate}
        \item [(a)] If $n \geq \binom{k+1}{2}$, then the minimum is given by a simple formula.
        \item [(b)] If $2k-1 \leq n < \binom{k+1}{2}$, then an iterative search is performed to find $m$ such that $n_{m-1} < n \leq n_m$. (However, the calculation of $n_m$ from $m$ in Lemma~\ref{lemThresholds} is not readily inverted and therefore does not yield a closed-form expression for the minimum.)  
        \item [(c)] Otherwise, if $n < 2k-1$, then symmetry (Lemma~\ref{lemMinSymmetry}) can be invoked and one of the steps (a) or (b)  will be applicable.
    \end{enumerate}
\end{remark}

\subsection{Bounds on the Minimum}

To address case (b) from~\cref{rmk:minimum_full}, we provide bounds on $a(n,k)$ that will help limit the search.
The first set of bounds were those proved by Katona~\cite{Katona}. The upper bound was later improved by Wegener~\cite{Wegener}. 

\begin{lemma}
    (\cite{Katona} Theorem 5, 6)
    For any positive integers $n,k$ with $n \geq 2k$, $\frac{\log n}{\log (en/k)} \frac{n}{k} \leq a_\mathrm{I}(n,k) \leq \ceil{\frac{\log (2n)}{\log (n/k)}} \frac{n}{k}$.
    \label{lem:Katona}
\end{lemma}

\begin{lemma}
  (\cite{Wegener} Main Theorem)
    For $n,k\in \mathbb{Z}^+$ with $n \geq 2k$, $a_\mathrm{I}(n,k) \leq \ceil{\log_q n} (q -1)$, where $q = \ceil{n/k}$.
    \label{lem:Wegener}
\end{lemma}

Now, we present and prove new bounds.

\begin{lemma}
    \label{lem:effBound2}
    When $2k\leq n < 1+ \binom{k+1}{2}$, $a_\mathrm{I}(n,k) \leq \ceil{\frac{-1 + \sqrt{8n-7}}{2}} $ and $a_\mathrm{II}(n,k) \leq\ceil{\frac{-1 + \sqrt{8n+1}}{2}}$.
\end{lemma}

\begin{proof}
    First, let us find a separating system with subsets of size $k'$, where $k' \leq k$. Observe that if $n \leq 1+\binom{k'+1}{2}$, then combining this with $n \geq 2k \geq 2k'$ as well as \cref{lemMinOrderingI,lemMinLargen} gives us $a_{\mathrm{I}}(n,k') \leq a_{\mathrm{I}}(1+\binom{k'+1}{2}, k') = k'$. Note that $n \leq 1+\binom{k'+1}{2} $ is equivalent to $(k')^2 + k' -(2n-2) \geq 0$. Letting $k^*$ be the smallest such $k'$ that satisfies this property, then we have that $a_{\mathrm{I}}(n,k^*) \leq k^* = \ceil{\frac{-1 + \sqrt{8n-7}}{2}}$.

    Finally, since $k^* \leq k$, then by \cref{lemMinOrderingI}, we have that $a_{\mathrm{I}}(n,k) \leq a_\mathrm{I}(n,k^*) = k^*$. Moreover, $a_{\mathrm{II}}(n,k) = a_{\mathrm{I}}(n+1,k) =\ceil{\frac{-1 + \sqrt{8n+1}}{2}} $.
\end{proof}

\begin{lemma}
\label{lem:modWeg}
    Suppose $n \geq 2k$, and let $q = \ceil{n/k}$. Then, $a_\mathrm{I}(n,k) \leq \ceil{(q-1) \log_q n}$.
\end{lemma}

\begin{proof}
    We follow a similar constructive procedure as Wegener~\cite{Wegener}: we begin with an empty matrix with $n$ columns, and in each step we add a block of $q-1$ rows. We assign a unique length-$(q-1)$ column pattern to each of the $q$ possible digits $\{0, 1, \dots, q-1\}$ using the $q-1$ standard basis vectors and the all-zero vector. We index the $n$ columns from $0$ to $n-1$. In the $t$-th step (block), the pattern assigned to a column is determined exactly by the $t$-th digit of its index in base-$q$. Because there are $n$ total columns, any specific base-$q$ digit appears at most $\ceil{n/q}$ times in any given position. Since $q = \ceil{n/k}$, this guarantees that any specific pattern (and thus any $1$ in the row block) is used at most $\ceil{\frac{n}{\ceil{n/k}}} \le k$ times, strictly satisfying the LCCP capacity constraint.
    
    Letting $m = \ceil{(q-1) \log_q n}, t = \floor{\log_q n}$, and $\epsilon = \log_q n -t$, we split the construction into $t$ full steps and a final partial step.
    Each row has at most $k$ ones, so after $t$ full steps (i.e., $t(q-1)$ rows), the columns are partitioned into $q^t$ possible groups, where the largest group of identical columns has size $s_t \leq \ceil{n/q^t}$, which will be separated in the final partial step.
    
    Viewing the number of rows $m$ as a budget, we observe that we have a remaining budget of:
    \begin{align*}
        m_\text{part} &= \ceil{(q-1) \log_q n} - t(q-1) \\
        &= \ceil{(q-1)(t+\epsilon)}-t(q-1)  \\
        &= \ceil{\epsilon(q-1)} \quad \text{rows}.
    \end{align*}
    We must show that $m_\text{part}$ rows is enough to distinguish all columns in the remaining groups, where the largest group has $s_t$ identical columns.
    To separate a group of $s_t$ columns, we can add $s_t-1$ new rows as follows: for $i \in \{1,2,\dots,s_{t}-1\}$, in row $i$ we add a $1$ at the $i$-th column of the group and zeros elsewhere.
    Thus, what remains is to show that $m_\text{part} \geq s_t-1$, i.e., 
    \begin{equation*}
        \ceil{\epsilon(q-1)} \geq \ceil{q^\epsilon}-1.
    \end{equation*}
    Define real functions $f,g$ such that $f(\epsilon) = 1+\epsilon(q-1) , g(\epsilon) = q^\epsilon$.
    Observe that for $\epsilon\in [0,1]$, $g$ is convex, and $f$ is the line segment connecting $g(0) = 1$ and $g(1) = q$.
    By convexity, $g(\epsilon) \leq f(\epsilon)$, so $q^\epsilon \leq 1+\epsilon(q-1)$, and:
    \begin{equation*}
        s_t -1 = \ceil{q^\epsilon}-1 = \ceil{q^\epsilon-1} \leq \ceil{\epsilon(q-1)} = m_\text{part}.
    \end{equation*}
    Finally, while this construction yields rows with weights strictly less than or equal to $k$ (including the final partial rows which have weight exactly $1$), \cref{lemAtMostEquivalence} guarantees that this is mathematically equivalent to the existence of a valid minimal draw matrix with exactly weight-$k$ rows. Thus, $a_\mathrm{I}(n,k) \le m$, concluding the proof.
\end{proof}

Now, we compare the bounds that we derived with each other, as well as with the bounds from~\cite{Katona,Wegener}. \cref{fig:bounds} shows plots of the Type-I minimum for $k = 20$ (i.e., $a_{\text{I}}(n,20)$) along with the bounds from \cref{lemLogBound,lemma_fixedk,lem:effBound2,lem:modWeg}, where the lower bounds are shown in \cref{fig:bounds_a} and the upper bounds are shown in \cref{fig:bounds_b}. Observe that the logarithmic bound from \cref{lemLogBound} is very loose when $n$ is large, and the same can be said of the bounds by Katona~\cite{Katona} and Wegener~\cite{Wegener}. Moreover, the bounds we derived are relatively tight, though finding an even tighter bound or approximation is an open problem.

\begin{figure}[htbp]
    \centering
    \begin{subfigure}[b]{0.48\linewidth}
        \centering
        \includegraphics[width=\linewidth]{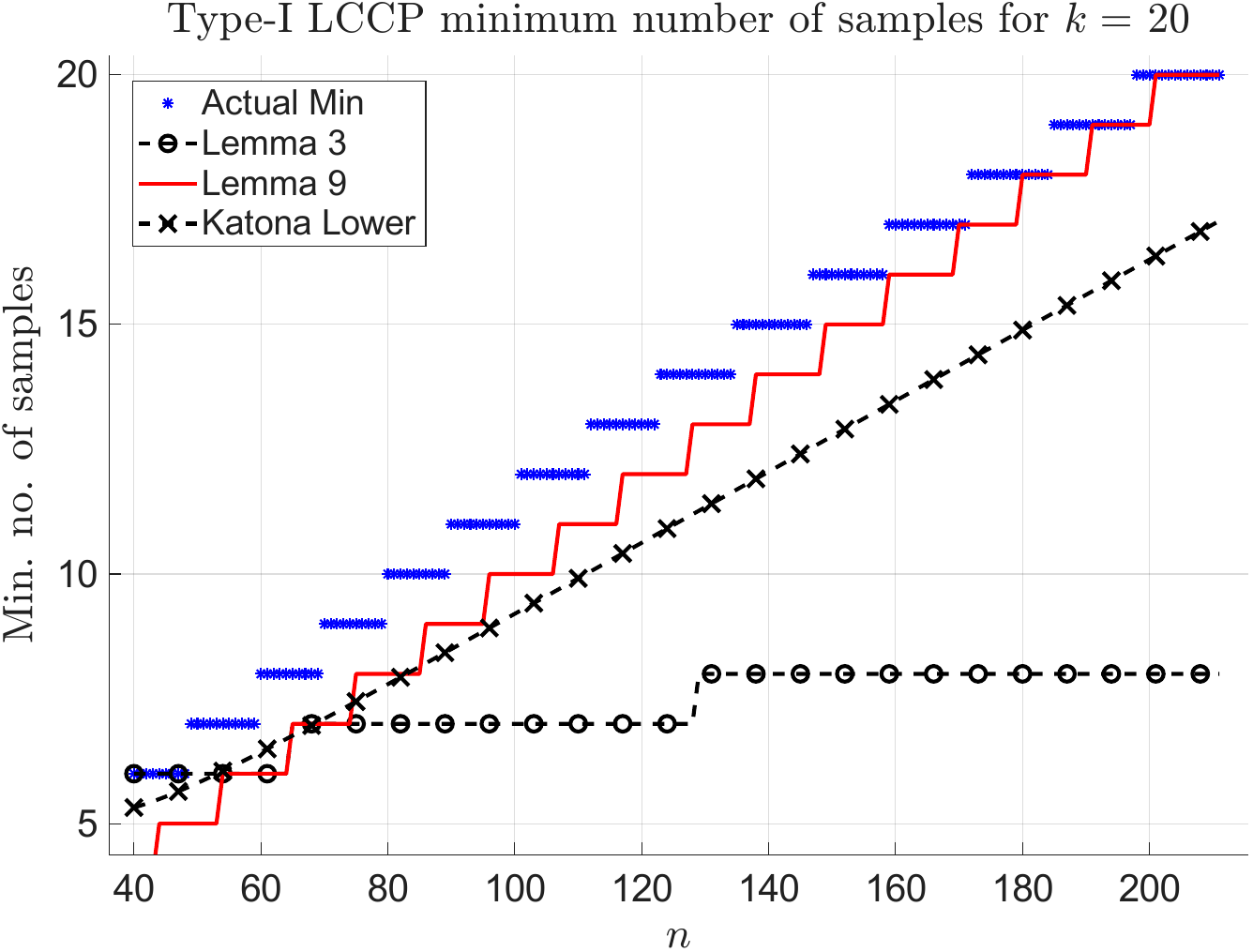}
        \caption{Comparison of lower bounds.}
        \label{fig:bounds_a}
    \end{subfigure}
    \begin{subfigure}[b]{0.48\linewidth}
        \centering
        \includegraphics[width=\linewidth]{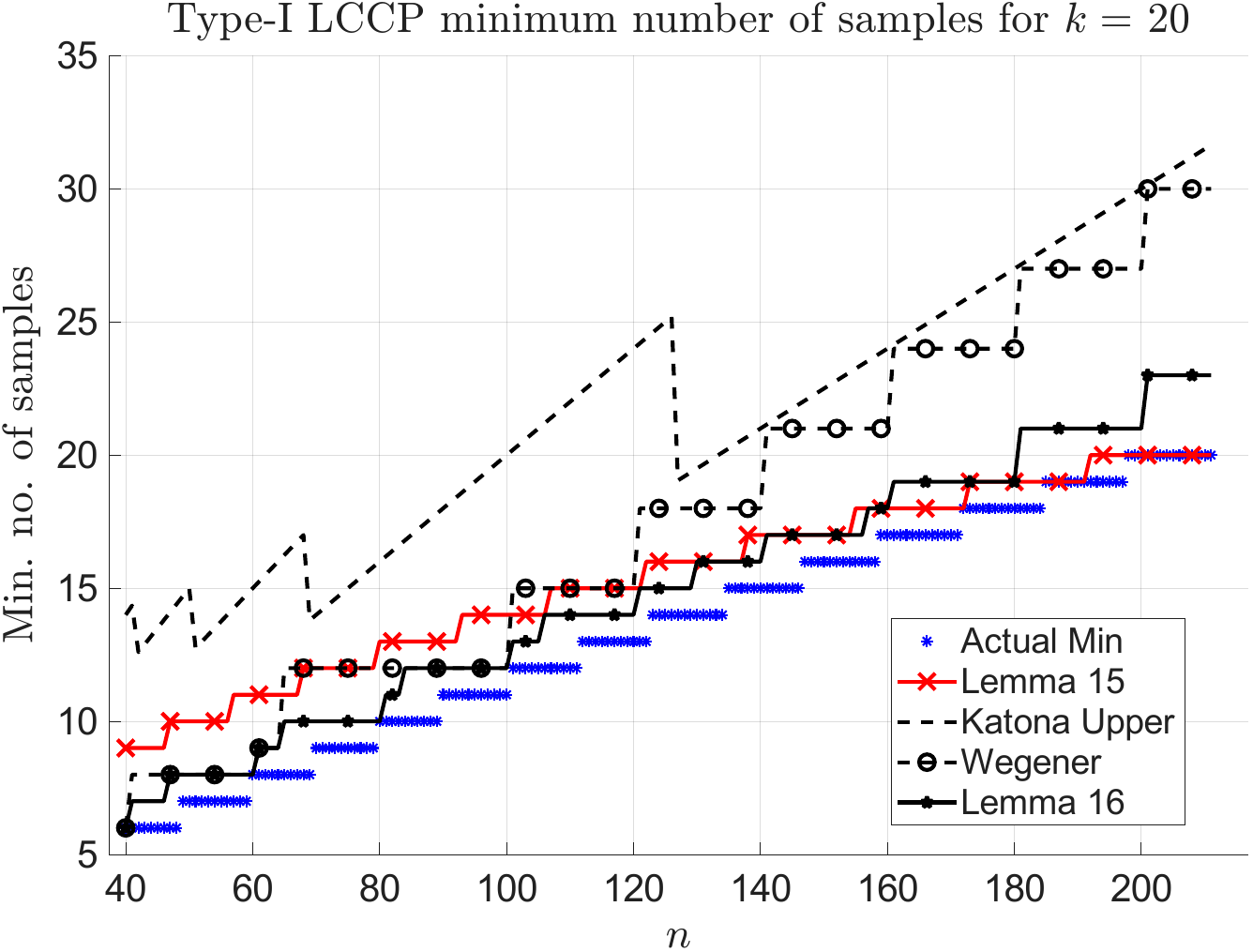}        
        \caption{Comparison of upper bounds.}
        \label{fig:bounds_b}
    \end{subfigure}
    
    \caption{Comparison of bounds on the smallest number of samples required for the $(n,20)$-Type-I LCCP, where (a) shows lower bounds and (b) shows upper bounds. Our bounds are given as solid lines.}
    \label{fig:bounds}
\end{figure}

\section{Partial Recovery}
In this section, we consider the \textit{partial recovery} case where we need to determine only $r$ of the coupons' labels for $r < n$, whereas previously we considered the full recovery case with $r = n$.

Let $Q(n,k,r)$ denote the first sample in which the image of $f$ can be determined for $r$ coupons; i.e., there exists a subset $J \subset [n]$ of size $r$ such that $f(j)$ can be determined for every $j \in J$. Accordingly, let $Q_\text{I}(n,k,r)$ and $Q_\text{II}(n,k,r)$ be the first such sample under the Type-I and II settings respectively, and let $a_\mathrm{I}(n,k,r)$ and $a_\mathrm{II}(n,k,r)$ denote their minimum possible values. Equivalently, $a_\mathrm{I}(n,k,r)$ (resp. $a_\mathrm{II}(n,k,r)$) is the minimum number of rows for a binary matrix with $n$ columns and weight-$k$ rows to have $r$ columns that appear exactly once in the matrix (resp. and are nonzero). Unlike with full recovery, the remaining $n-r$ columns can be anything, including zero or identical to each other.


We begin with the following observation.
\begin{lemma}
    For all $n, k,r \in \mathbb{Z}^+$ such that $n > k$ and $n > r$, $a_\mathrm{I}(n,k,r) \leq a_\mathrm{II}(n,k,r)$. 
\end{lemma}
\begin{proof}
    This is immediate from the observation that the property that a matrix needs to satisfy in the Type-II setting is stronger than that required for Type-I. 
\end{proof}

Since we only need to deduce the labels of $r$ coupons, we expect that the number of samples required will be less than that required to deduce the labels of all coupons, especially when $r$ is much smaller than $n$. This can be seen through an analog of \cref{lemLogBound}. 

\begin{lemma}
    \label{lemPartialLogBound}
    For $n,k,r \in \mathbb{Z}^+$ such that $n > r$ and $n > k$, $$a_\mathrm{II}(n,k,r) \geq a_\mathrm{I}(n,k,r) \geq a_\mathrm{II}(r, \ceil{r/2}) = \ceil{\log_2 (r+1)}.$$
\end{lemma}

\begin{proof}
    Since we require $r$ columns to appear exactly once, this means that any suitable draw matrix with $m$ rows has at least $r+1$ different columns, which means $2^m \geq r+1$, implying $m \geq \ceil{\log_2(r+1)}$. The same can be said about the Type-II setting, since we could have $r$ distinct nonzero columns and allow the rest to be zero. 
\end{proof}

Using this lemma as well as \cref{lemMinOrderingII,lemMinOrderingI}, we can find $a(n,k,r)$. 

\begin{theorem}
\label{thmMinPartial}
Let $n,k,r$ be positive integers such that $n > r$ and $n > k$, and let $m^* =  a_\mathrm{II}(r, \ceil{r/2}) =\ceil{\log_2(r+1)}$, i.e., the minimum number of samples to separate $r$ coupons for the full recovery case. 
Defining $s:= k - (n-r)$, we observe that the partial recovery minimum is determined by the ability to allocate ones into $r$ target columns. Namely, we have the three regimes:

\begin{enumerate}
    \item [(1)]\textbf{Starved Regime}: If $k \leq \ceil{r/2}$, then  
    $$a_\mathrm{I}(n,k,r) = a_\mathrm{II}(n,k,r) = a_\mathrm{II}(r,k).$$
    
    \item [(2)] \textbf{Flooded Regime}: If $s > \ceil{r/2}$, then  
    $$a_\mathrm{I}(n,k,r) = a_\mathrm{I}(r, s),$$  $$a_\mathrm{II}(n,k,r) = a_\mathrm{II}(r, s).$$
    
    \item  [(3)] \textbf{Optimal Regime}: If $s \leq \ceil{r/2} < k$, then 
    $a_\mathrm{I}(n,k,r)$ and $a_\mathrm{II}(n,k,r)$ achieve the absolute minimum $m^*$, with the exception of two extreme cases requiring $m^* + 1$ samples:
    \begin{itemize}
        \item $r = 2^{m^*} - 1$.
        \item $r = 2^{m^*} - 3$ and $s = \ceil{r/2}$.
    \end{itemize}
\end{enumerate}
\end{theorem}

\begin{proof}
    For any draw matrix separating $r$ target coupons, we can represent it as a block matrix $[M|M'] \in \mathbb{F}_2^{m \times n}$, where $M \in \mathbb{F}_2^{m \times r}$ corresponds to the target coupons and $M' \in \mathbb{F}_2^{m \times (n-r)}$ to the remaining coupons. In every row, the $k$ ones must be distributed between $M$ and $M'$. Since the row weight of $M'$ cannot exceed $n-r$, the row weight $w$ of $M$ is constrained by $k - (n-r) \leq w \leq k$. Let $s = k - (n-r)$.

    We prove the bounds by analyzing three regimes based on the capacity to allocate ones to $M$.

    \textbf{Case 1 ($k \leq \ceil{r/2}$):} 
    Suppose a matrix $M \in \mathbb{F}_2^{m \times r}$ corresponds to a separating and covering system for $[r]$ with weight $k$. The block matrix $[M|0] \in \mathbb{F}_2^{m \times n}$ separates and covers the first $r$ coupons, so $a_\mathrm{I}(n,k,r) \leq a_\mathrm{II}(n,k,r) \leq m$. This holds for $m = a_\mathrm{II}(r,k)$.
    
    To prove the reverse inequality, let $m = a_\mathrm{II}(n,k,r)$ and consider a minimal partial recovery matrix $[M|M'] \in \mathbb{F}_2^{m \times n}$. We extract the submatrix $M \in \mathbb{F}_2^{m \times r}$, which successfully separates and covers the $r$ target coupons. Since the row weight of $M$ is at most $k' \leq k$, \cref{lemAtMostEquivalence} gives $a_\mathrm{II}(r,k') \leq m$. Moreover, since $k' \leq k \leq \ceil{r/2}$, \cref{lemMinOrderingII} yields $a_\mathrm{II}(r, k) \leq a_\mathrm{II}(r,k') \leq m$, establishing $a_\mathrm{II}(r,k) \leq a_\mathrm{II}(n,k,r)$. If we are evaluating the Type-I minimum and $M$ contains an all-zero column, we can swap it with a nonzero column from $M'$ to form a modified $m \times r$ matrix. As we are only extracting an $m \times r$ full recovery matrix to satisfy the lower bound for $a_\mathrm{II}(r,k)$, we do not need to worry about creating duplicates in the remaining $M'$ submatrix; we only need our extracted $r$ columns to be distinct and nonzero.

    \textbf{Case 2 ($s > \ceil{r/2}$):} 
    Here, $2s > r$ and $r > 2(n-k)$.
    For the lower bound, let $m = a_\mathrm{II}(n,k,r)$ and consider any minimal matrix $[M|M']$. The submatrix $M'$ has a maximum row weight of $n-r$, meaning $M$ has a minimum row weight of $s$. We extract $M$ and consider its component-wise complement $M^\mathrm{c}$, which has $r$ distinct columns and a maximum row weight of $r-s$. Since $s > \ceil{r/2}$, we strictly have $r \geq 2(r-s)$. Thus, the existence of $M^\mathrm{c}$ and \cref{lemAtMostEquivalence} implies $a_\mathrm{I}(r, r-s) \leq m$. Applying \cref{lemMinSymmetry,lemTypeMin}, we obtain $a_\mathrm{II}(r,s) = a_\mathrm{I}(r,s) = a_\mathrm{I}(r,r-s) \leq m$. Therefore, $a_\mathrm{II}(n,k,r) \geq a_\mathrm{II}(r,s)$. The Type-I bound follows similarly.
    
    To prove the reverse inequality, we construct $[M| 1_{m\times (n-r)}]$ where $M \in \mathbb{F}_2^{m \times r}$ has weight $s$ and $m = a_\mathrm{II}(r,s)$. Since $2s > r$, \cref{lemTypeMin,lemMinSymmetry} yield $a_\mathrm{II}(r, s) = a_\mathrm{I}(r,s) = a_\mathrm{I}(r, r-s)$. We construct $M$ by generating a matrix $M_\text{orig}$ using the procedure from \cref{lemThresholds} for weight $r-s$ and taking its component-wise complement. For this to form a valid Type-II partial recovery matrix, $M$ must avoid both the all-zeros and all-ones vectors. $M_\text{orig}$ natively avoids the all-zeros vector (ensuring $M$ avoids all-ones) since $r-s > 0$. Furthermore, $M_\text{orig}$ avoids the all-ones vector (ensuring $M$ avoids all-zeros) because utilizing it would require full capacity, meaning $r-s = 2^{m-1} - 1$ and $r = 2^m - 2$. This implies $s = 2^{m-1}-1$, yielding $r = 2s$, which explicitly contradicts our regime condition $2s > r$.

    \textbf{Case 3 ($s \leq \ceil{r/2} < k$):} 
    By \cref{lemPartialLogBound}, $a_\mathrm{I}(n,k,r)$ and $a_\mathrm{II}(n,k,r)$ are at least $m^*$. We prove the upper bounds by allocating exactly $\ceil{r/2}$ ones to $M$ and the remaining $k-\ceil{r/2} \leq n-r$ ones to $M'$ per row. 
    
    If $r$ is even ($r \leq 2^{m^*} - 2$), we populate $M$ with $r/2 \leq 2^{m^*-1} - 1$ complementary pairs of vectors in $\mathbb{F}_2^{m^*}$, strictly avoiding both the all-zeros and all-ones vectors. This ensures every row in $M$ has exactly $\ceil{r/2}$ ones. We then allocate exactly $k - \ceil{r/2}$ ones to $M'$ in every row by assigning $k - \ceil{r/2}$ columns of $M'$ to be all-ones vectors, and the remaining columns of $M'$ to be all-zeros vectors. This validly achieves $m^*$.

    If $r$ is odd, we evaluate three sub-cases:
    \begin{itemize}
        \item \textbf{$r = 2^{m^*} - 1$:} We require the $r$ target columns to appear exactly once in the matrix. Because $r = 2^{m^*} - 1$, $M$ must contain every non-zero vector in $\mathbb{F}_2^{m^*}$. Since $s \leq \ceil{r/2} < k$, we have $k - \ceil{r/2} \ge 1$, meaning we must allocate at least one $1$ to $M'$ in every row. Thus, $M'$ must contain at least one non-zero column. Since $M$ already exhausts all possible non-zero vectors of length $m^*$, this column in $M'$ will inevitably collide with a column in $M$, violating the requirement that the $r$ target columns appear exactly once. Thus, $m^*$ rows are strictly insufficient.
        
        To construct a valid matrix using $m^*+1$ rows, we populate the first $m^*$ rows of $M$ with all $2^{m^*}-1$ non-zero vectors, and assign exactly $k - \ceil{r/2}$ columns of $M'$ to be all-ones vectors (with the remaining columns of $M'$ being all-zeros). This creates a collision between the all-ones column in $M$ and the all-ones columns in $M'$. To resolve this, we append the $(m^*+1)$-th row: we assign a $1$ to the colliding columns in $M'$, a $0$ to the all-zeros columns in $M'$, and distribute the remaining $\ceil{r/2}$ ones arbitrarily among the columns of $M$, ensuring the column in $M$ that was previously all-ones receives a $0$. This guarantees the all-ones column in $M$ becomes $(1, \dots, 1, 0)^T$, strictly separating it from the $M'$ columns which become $(1, \dots, 1, 1)^T$. The all-zeros columns in $M'$ become $(0, \dots, 0, 0)^T$, which cannot collide with $M$ since every column in $M$ contains at least one $1$ in the first $m^*$ rows.
        
        \item \textbf{$r = 2^{m^*} - 3$:} Here, $\ceil{r/2} = 2^{m^*-1}-1$. If $s = \ceil{r/2}$, then $k - \ceil{r/2} = n-r$, meaning $M'$ must consist entirely of all-ones vectors. To avoid colliding with $M'$, $M$ cannot use the all-ones vector. Thus, $M$ must use every vector in $\mathbb{F}_2^{m^*}$ except three (including $\mathbf{1}$). The three excluded vectors (e.g., $\mathbf{1}$, $v$, and $\overline{v}$) naturally remove exactly two ones from every single row. Consequently, every row in $M$ is forced to have exactly $2^{m^*-1} - 2 = \ceil{r/2} - 1$ ones. Adding the ones from $M'$ yields a total row weight of $(\ceil{r/2} - 1) + (k- \ceil{r/2}) = k-1$, which strictly violates the weight $k$ constraint. This parity conflict occurs regardless of whether the all-zeros vector is permitted, meaning both Type-I and Type-II strictly require $m^*+1$ samples. Constructing a valid matrix requires inserting an additional row as done previously, yielding $m^*+1$ rows. If $s < \ceil{r/2}$, the all-ones constraint on $M'$ is relaxed, and $m^*$ rows suffice.
        
        \item \textbf{$r \leq 2^{m^*} - 5$:} We populate $M$ with vectors $v_1 = [0,1,1,\dots,1]^T$, $v_2 = [1,0,1,\dots,1]^T$, $v_3 = [1,1,0,\dots,0]^T$, followed by $\ceil{r/2}-2$ complementary pairs. This leaves $2^{m^*}-8$ nonzero vectors, which is sufficient to form the required $2^{m^*-1}-4 \geq (r+1)/2-2$ pairs, achieving $m^*$.
    \end{itemize}
\end{proof}

\begin{remark}
    Thus, for the minimum number of samples, every case in the partial recovery scenario reduces to the full recovery scenario with different parameters. It can be seen that the Type-I minimum is almost identical to the Type-II minimum, where the only difference lies in the flooded regime where there is a ``surplus'' of ones that force many columns of the draw matrix to be all-ones. In this case, the allowance of an all-zero column could potentially give the Type-I setting an ``advantage'' over Type-II. 
    Moreover, if the problem setting is changed so that each sampled subset can have a size of \textit{at most} $k$, then the absolute minimum $m^*$ can be achieved in both the flood and optimal regimes.
\end{remark}

We conclude the section with the following result for when each sample has \textit{at most} $k$ coupons rather than exactly $k$ coupons.
\begin{corollary}
    Let $b(n,k,r)$ be the minimum number of samples of size \textit{at most} $k$ such that the label of at least $r$ coupons can be deduced. Then:
    \begin{align}
        b(n,k,r) = \begin{cases}
            a_\mathrm{II}(r,k), & k \leq \ceil{r/2} \\
            m^*, & k > \ceil{r/2}
        \end{cases}.
    \end{align}
\end{corollary}

\begin{proof}
    The first case follows by \cref{lemAtMostEquivalence}. 
    When $k >\ceil{r/2}$, we construct a suitable draw matrix by simply choosing any $r$ nonzero columns from $\mathbb{F}_2^{m^*}$ and then filling the rest of the columns with the all-zeros vector.
\end{proof}

\section{Conclusion}
\label{sec:conclusion}


We introduced the LCCP, an extension of the classical CCP where  samples of $n$ labeled coupons, taken uniformly with replacement,  yield an unordered set of $k$ coupons and an unordered set of their $k$ associated distinct labels. 

For the full recovery scenario, we provided a complete characterization of the minimum required number of samples, $a_\mathrm{I}(n,k)$ and $a_\mathrm{II}(n,k)$, for the Type-I and Type-II settings, corresponding to a known and unknown label sets, respectively. By translating the subset drawings into a binary draw matrix formulation, we derived the exact threshold $n_m$ representing the maximum number of coupons whose labels can be deduced in $m$ samples. Furthermore, we developed explicit block-matrix constructions that allowed us to establish upper and lower bounds that are significantly tighter than those previously established by Katona~\cite{Katona} and Wegener~\cite{Wegener}, particularly in regimes where $n$ is much larger than $k$.

We also extended our deterministic framework to the partial recovery variation of the LCCP, where only a subset of $r$ target coupons must be identified. We categorized the partial recovery problem into three distinct regimes: Starved, Flooded, and Optimal. We demonstrated that in the optimal regime, the required number of samples natively achieves the unconstrained theoretical minimum $m^* = \ceil{\log_2(r+1)}$, barring two highly dense edge cases that strictly require one additional sample due to parity and collision constraints.


Several interesting open directions remain for future work. One challenge is discovering an exact, closed-form expression for the minimum number of samples $a(n,k)$ in the intermediate regime where $2k \leq n < \binom{k+1}{2}$, avoiding a computational search over $m$. Additionally, while our constructed bounds significantly improve upon prior literature, finding even tighter bounds or analytical approximations remains an open problem. Finally, one can investigate practical applications of these results to Chaumian threshold mixes and other anonymity protocols by extension of the model. For example, by relaxing the definition of the mapping $f$ to be a general function (not necessarily a bijection), one can address the case of one-to-many communication patterns.

\begin{outline}
    \section{References}
\end{outline}


\newpage
\phantom{ }
\appendix
\markboth{Appendix}{Appendix}
\noindent

To prove \cref{prop_almost}, we start by considering simpler cases.

\begin{proposition}
\label{prop_cyclic}
    \cref{prop_almost} holds when $m > w, k \leq w$, and $m$ and $w$ are coprime.
\end{proposition}

\begin{proof}
    Let $v\in \mathbb{F}_2^m$ be the binary $m$-tuple whose first $w$ entries are $1$, and for $i \in \{0,\dots,m-1 \}$, let $v^{(i)}$ be the tuple obtained via $iw$ cyclic shifts of $v$ (e.g., $v^{(0)} = v$). Since $m$ and $w$ are coprime, the map $i \mapsto iw \pmod{m}$ is a bijection, and so $\{v^{(i)}\}_{i=0}^{m-1}$ contains all cyclically shifted versions of $v$; the claim is that $M$ can be constructed using $v^{(0)},\dots,v^{(n-1)}$ as its columns.

    Each row index can be viewed as a residue class modulo $m$, and each $v^{(i)}$ can be viewed as a consecutive sequence of residue classes corresponding to the rows where $v^{(i)}$ has a $1$. In particular, $v^{(i)}$ has $1$'s in positions $(i-1)w,\dots,iw-1 \pmod{m}$; more importantly, each $v^{(i)}$ is a "continuation" of the previous tuple, e.g., $v^{(i-1)}$ has ones in positions $(i-2)w,\dots,(i-1)w-1 \pmod{m}$. This allows us to "chain" the $v^{(i)}$'s together to form sequences of complete residue systems (i.e., $\{0,1,\dots,m-1  \pmod{m}\}$), where each complete residue system represents a $1$ in every row.
    With $n = \floor{\frac{mk}{w}}$, the number of complete residue systems formed by $v^{(0)},\dots,v^{(n-1)}$ is $\floor{\frac{nw}{m}}$, which is either $k-1$ or $k$.
    Indeed, if $n = \floor{\frac{mk}{w}}$, we have $nw \leq mk$ and $(n+1)w> mk \Rightarrow mk - nw < w$, and so:

    \begin{equation*}
        \floor{\frac{nw}{m}} = \floor{\frac{mk}{m} - \frac{mk-nw}{m}}\geq \floor{k-\frac{w}{m}}= k-1.
    \end{equation*}

    This tells us that there are $k-1$ complete residue systems and a remainder of $r = nw-(k-1)m = nw-km+m >m-w$ consecutive residue classes. Therefore, $M = \begin{bmatrix}
        v^{(0)} & \dots & v^{(n-1)}
    \end{bmatrix} \in \mathbb{F}_2^{m\times n}$ will have $r > m-w$ rows of weight $k$ while the rest have weight $k-1$ as desired.

\end{proof}

If $m$ and $w$ are not coprime, the proposition can be slightly generalized as follows:
\begin{proposition}
\label{prop_cyclic_2}
    \cref{prop_almost} holds when $m > w$ and $k \leq \frac{w}{g}$, where $g = \gcd(m,w)$.
\end{proposition}
\begin{proof}
    Let $m' = \frac{m}{g}$ and $w' = \frac{w}{g}$, which means that $\gcd(m',w') = 1$. With $k \leq w'$, then by \cref{prop_cyclic} there exists $M' \in \mathbb{F}_2^{m'\times n}$ whose columns are distinct and have weight $w'$, and more than $m'-w'$ rows have weight $k$ while the rest have weight $k-1$. Then, by replacing each $1$ in $M'$ with an all-ones $g$-tuple and each $0$ with an all-zeros $g$-tuple, we obtain $M \in \mathbb{F}_2^{m\times n}$ with the desired properties.
\end{proof}

We can finally prove the main proposition:

\begin{proof}
    Denote the set of weight-$w$ tuples as $V_w$. First, observe that if $k = \binom{m-1}{w-1}$, then $n = \binom{m}{w}$, which means we can choose the entirety of $V_w$ as columns for $M$, so let us assume that $k < \binom{m-1}{w-1}$.

    Consider the group action of cyclic shifts on $V_w$, where $V_w$ can be partitioned into (possibly unequally-sized) orbits invariant under cyclic shifts with order dividing $m$. Each orbit will have an equal number of ones per row (via symmetry), where the largest orbit will have $w$ ones per row. One such largest orbit is the orbit of $v \in V_w$, where the first $w$ entries of $v$ are $1$.

    The sum of the number of ones per row across all orbits is $\binom{m-1}{w-1}$, and since $k < \binom{m-1}{w-1}$, then it is possible to obtain at least $k$ ones per row through insertion of orbits into $M$. Moreover, since each orbit has at most $w$ ones per row, then it is possible to obtain at least $k-w$ but no more than $k$ ones per row through insertion of orbits except that of $v$. We do exactly this; letting $n'$ be the current number of columns and $k'$ be the current number of $1$'s per row in $M$, we have $k - k' \leq w$ and $n' w = mk' \Rightarrow n' = \frac{mk'}{w}$.

    The last step is to insert a subset of the orbit of $v$ with $\floor{\frac{m(k-k')}{w}} $ columns and $k-k'$ ones per row so that $M$ will have $n' +\floor{\frac{m(k-k')}{w}} = \frac{mk'}{w}  + \floor{\frac{m(k-k')}{w}}=\floor{\frac{mk}{w}}$ columns and $k$ ones in more than $m-w$ of the rows (and $k-1$ ones in each of the rest).

    Letting $v^{(i)}$ for $i \in \mathbb{Z}$ be the tuple obtained via $i$ cyclic shifts of $v$ and $g = \gcd(m,w)$, observe that the orbit of $v$ can be further partitioned into $m/g$ smaller "suborbits" of the form $V'_i = \{v^{(i + jg)}: j\in \mathbb{Z}^+ \}$ for $i \in \{0,\dots,g-1\}$ (if $g = 1$, then nothing happens), where each $V_i'$ has cardinality $g$ and exactly $w/g$ ones per row.

    If $k-k'$ is a multiple of $w/g$ we can simply insert a whole number of suborbits into $M$, so suppose that $k-k' = qw/g + r$ for some $r < w/g, q< g-1$. Barring $V_0'$, we can use any $q$ suborbits and form the remaining part of $M$ with $V_0'$ using \cref{prop_cyclic_2}. Indeed, $M$ will have a total of $n' + \frac{m}{g}q + \floor{\frac{rm/g}{w/g}} = n' + \floor{\frac{m/g \cdot (qw/g+r)}{w/g}}= n' + \floor{\frac{m(k-k')}{w}} = n$ columns of weight $w$, and will have rows where more than $m-w$ of them have weight $k$, and the rest have weight $k-1$.
\end{proof}

\begin{algorithm}
\caption{Construction of maximal LCCP draw matrix given number of samples $m$ and row weight $k$}
\label{alg:maximal_lccp}
\SetKwInOut{Input}{Input}
\SetKwInOut{Output}{Output}

\Input{Integers $m$ (samples) and $k$ (row weight).}
\Output{Maximal integer $n$ and binary matrix $M \in \mathbb{F}_2^{m \times n}$.}

\tcp{Step 1: The Knapsack Calculation (Determine max n)}
Initialize $n \gets 0$, Budget $B \gets mk$, $w \gets 1$\;
Initialize array $A$ where $A_w$ stores count of columns of weight $w$\;
\While{$B \ge w$}{
    $MaxCols \gets \binom{m}{w}$\;
    $Cost \gets w \times MaxCols$\;
    \eIf{$B \ge Cost$}{
        $A_w \gets MaxCols$\;
        $n \gets n + MaxCols$\;
        $B \gets B - Cost$\;
        $w \gets w + 1$\;
    }{
        $A_w \gets \lfloor B / w \rfloor$\;
        $n \gets n + A_w$\;
        $r \gets w - 1$ \tcp*[r]{Last full weight class}
        \textbf{break}\;
    }
}
$w_{part} \gets r + 1$ \tcp*[r]{Partial weight class}

\tcp{Step 2: Construct Full Layers}
$M \gets$ empty matrix with $m$ rows\;
\For{$i \gets 1$ \KwTo $r$}{
    Generate all $\binom{m}{i}$ vectors of weight $i$\;
    Append vectors to $M$ as columns\;
}

\tcp{Step 3: Construct Partial Layer (Almost Regular)}
Generate base vector $v$ of length $m$ with $1$s in first $w_{part}$ indices\;
Generate cyclic shifts $v^{(0)}, v^{(1)}, \dots$\;
Select first $A_{w_{part}}$ distinct vectors (prioritizing largest orbits if $\gcd(m, w_{part}) > 1$)\;
Append vectors to $M$\;

\tcp{Step 4: Deficiency Correction}
$D \gets \{ \text{indices of rows with weight } k-1 \}$\;
$d \gets |D|$\;
\If{$d > 0$}{
    Find column $c$ in $M$ with weight $r+2-d$ having $0$s in all rows in $D$\;
    \For{$row \in D$}{
        $M[row, c] \gets 1$ \tcp*[r]{Flip 0 to 1}
    }
}

\KwRet $n, M$\;
\end{algorithm}

\end{document}